\documentclass[3p,authoryear]{elsarticle}
\usepackage[utf8]{inputenc}
\usepackage{caption}
\usepackage{subcaption}
\usepackage{graphicx}
\usepackage{soul, color}
\usepackage{amsfonts}
\usepackage{calrsfs}
\usepackage{hyperref}
\DeclareMathAlphabet{\pazocal}{OMS}{zplm}{m}{n}
\usepackage{amsmath}
\usepackage{bbm}
\DeclareMathOperator*{\argmax}{arg\,max}

\usepackage{bm}
\usepackage{amsthm}
\usepackage{upgreek}
\usepackage{empheq}

\usepackage{hyperref}
\newtheorem{theorem}{Theorem}

\newtheorem{definition}{Definition}
\newtheorem{lemma}{Lemma}
\begin{document}
% Put here what will go to headers as author
\title{A Max Pressure Algorithm for Traffic Signals Considering Pedestrian Queues}

% TODO: add macros for easier formatting of \author.

\author[1]{Hao Liu\corref{cor1}}
\ead{hfl5376@psu.edu}
\author[1]{Vikash V. Gayah}
\ead{gayah@engr.psu.edu}
\author[2]{Michael Levin}
\ead{mlevin@umn.edu}
\cortext[cor1]{Corresponding author}
\address[1]{The Pennsylvania State University, University Park, PA, United States}
\address[2]{University of Minnesota, Minneapolis, MN, United States}

\begin{abstract}
This paper proposes a novel max-pressure (MP) algorithm that incorporates pedestrian traffic into the MP control architecture. Pedestrians are modeled as being included in one of two groups: those walking on sidewalks and those queued at intersections waiting to cross. Traffic dynamics models for both groups are developed. Under the proposed control policy, the signal timings are determined based on the queue length of both vehicles and pedestrians waiting to cross the intersection. The proposed algorithm maintains the decentralized control structure, and the paper proves that it also exhibits the maximum stability property for both vehicles and pedestrians. Microscopic traffic simulation results demonstrate that the proposed model can improve the overall operational efficiency---i.e., reduce person travel delays---under various vehicle demand levels compared to the original queue-based MP (Q-MP) algorithm and a recently developed rule-based MP algorithm considering pedestrians. The Q-MP ignores the yielding behavior of right-turn vehicles to conflicting pedestrian movements, which leads to high delay for vehicles. On the other hand, the delay incurred by pedestrians is high from the rule-based model since it imposes large waiting time tolerance to guarantee the operational efficiency of vehicles. The proposed algorithm outperforms both models since the states of both vehicles and pedestrians are taken into consideration to determine signal timings.
\end{abstract}

\begin{keyword} 
Max Pressure algorithm; Pedestrians; Decentralized traffic signal control; Maximum stability
\end{keyword}

\maketitle

\section{Introduction}
Real-time decentralized signal control for large-scale urban networks has attracted intensive research interest in the past decades. Various methods have been developed, including SCOOT \citep{hunt1981scoot}, OPAC \citep{gartner1983opac}, UTOPIA \citep{mauro1990utopia}, SCATS \citep{lowrie1990scats} and RHODES \citep{mirchandani2001real}. More recently, max-pressure (MP) has become a very popular decentralized traffic signal control algorithm. This algorithm was originally developed to address the packet scheduling problem in wireless communication networks \citep{tassiulas1990stability}, but was applied to traffic signal control in \citep{varaiya2013max}. Since then, many MP-based traffic signal control algorithms \citep{kouvelas2014maximum, gregoire2014capacity, xiao2014pressure, le2015decentralized, wu2017delay, li2019position, rey2019blue, levin2020max, mercader2020max, dixit2020simple, liu2022novel, xu2022integrating, wang2022learning, liu2023total} have been proposed as a result of its ease of implementation, prominent control performance, and fast computational speed. Additionally, certain MP algorithms share the following desirable properties: no need for external demand information and \emph{maximum stability}. The latter property indicates that the MP algorithm can accommodate any demand that can be served by an admissible control strategy. A detailed review of MP-based signal control algorithms can be found in \citep{levin2023max}.

The MP algorithms determine signal timings for each intersection based on the \emph{pressure} of all phases, which is calculated based on the values of a selected metric---such as queue length, travel time and travel delay---measured for all local movements at the intersection. In general, the MP control algorithm aims at keeping the overall pressure low to maintain the operational efficiency at the intersection, and the pressure of each phase reflects the pressure value that can be potentially ``released'' from the associated intersection by the phase. Therefore, the MP algorithms usually tend to allocate green time in favor of the phase with the highest pressure. 

According to this basic working mechanism, it is easy to imagine that the control performance of an MP algorithm is highly dependent on the selected metric for the pressure computation. Therefore, the first aspect that the MP-related studies in the literature have mainly concentrated on is the selection of the metric. The queue length is used as the metric in \citep{varaiya2013max}. Note that, since point queue models were employed for traffic dynamics modeling in \citep{varaiya2013max}, the queue length is equivalent to the number of vehicles of a given link. Thus, for simplicity and consistency with what is used in the MP literature, the term ``queue" is referred to as the number of vehicles in this paper. However, this metric does not consider the position, moving status of vehicles and maximum occupancies of links, which can affect the control performance significantly. Various MP-based algorithms have been proposed to address these drawbacks. For example, the maximum occupancies of links were taken into consideration in \citep{gregoire2014capacity, xiao2014pressure}. A position-weighted MP algorithm was proposed in \citep{li2019position}, which incorporates each vehicle's location on the link into the pressure calculation. Compared to queue length, travel time and travel delay are inherently influenced by the maximum occupancy of links, and travel time-based \citep{mercader2020max} and delay-based \citep{dixit2020simple, liu2022novel, liu2023total} MP algorithms have also been developed and shown to outperform queue-based MP approaches.

In addition to the metric, another pressing area of interest with respect to MP algorithms is its practicability. Many MP-based algorithms \citep{varaiya2013max, gregoire2014capacity, liu2022novel, li2019position} update phases at a fixed frequency in an arbitrary order, which could be a barrier for implementation since doing so may confuse travelers who expect regular cyclic structure. MP algorithms with fixed signal sequences have been developed to address this issue \citep{levin2020max, mercader2020max}. 

Another practical concern for the aforementioned MP algorithms is that they do not distinguish vehicle types in a mixed traffic flow environment, and this treatment could diminish the overall operational efficiency. An MP algorithm considering public transit signal priority in networks with exclusive bus lanes was proposed in \citep{xu2022integrating}. In addition to bus rapid transit, pedestrians are another important component of urban transportation networks. However, incorporating pedestrians into the MP control architecture has received very little attention in the research literature. \citet{xu2023ped} proposed an MP algorithm that considers pedestrian access at intersections by using a rule-based approach. Specifically, when the waiting time of the first-arrived pedestrian exceeds a predefined threshold value, the traffic signal provides protected phase to the corresponding pedestrian movement. Otherwise, the signals are controlled by a regular queue-based MP algorithm. Although the maximum stability of the proposed model was proved for vehicles, the stable region is dependent on the predefined threshold for pedestrian's waiting times. Also, the stability for pedestrians in the network is not clear. In addition, the proposed method did not consider the yielding behavior of right-turn vehicles to pedestrians during permissive phases, which is a common signal design in urban transportation networks.

In order to partially bridge these gaps, this paper proposes an MP algorithm that incorporates pedestrians into the pressure calculation while maintaining the MP control structure. The contributions of this paper include: 1. develop a traffic dynamics model for pedestrians; 2. propose a queue-based MP algorithm that considers both vehicle and pedestrian queues; 3. demonstrate the proposed model inherits the maximum stability property for both vehicles and pedestrians; 4. show the proposed model can outperform the queue-based MP proposed in \citep{varaiya2013max} and the rule-based model proposed in \citep{xu2023ped} under various traffic conditions.

The remainder of this paper is organized as follows. Section \ref{sec:dynamics} describes the traffic dynamics models for both vehicles and pedestrians. Then, the novel MP algorithm that incorporates pedestrian states is proposed, and the maximum stability is proved in Section \ref{sec:pqmp}. Section \ref{sec:simulation} shows the results of microscopic traffic simulations, which demonstrate that the proposed model can outperform two baseline models under various vehicle demand levels. Finally, concluding remarks are provided in Section \ref{sec:conclusion}. 

\section{Traffic dynamics}\label{sec:dynamics}
First, we define the required notations used for the modeling of traffic dynamics in this paper. We consider a network consisting of a set of intersections $\pazocal{N}$ and a set of directional links $\pazocal{L}$. A link that connects a pair of adjacent intersections, $i$ and $j$, is denoted by $(i,j)$. A movement from link $(i,j)$ to link $(j,k)$ is denoted by $(i,j,k)$. Links are classified as either entry links $\pazocal{L}_\mathrm{enter}$ from where vehicles enter the network, exit links $\pazocal{L}_\mathrm{exit}$ from where vehicles exit the network, and internal links $\pazocal{L}_{\mathrm{in}}$ that simply connect intersections. Also, let $\pazocal{U}^i$ ($\pazocal{D}^i$) indicate the set of upstream (downstream) intersections of intersection $i$ if $i$ is not an entry (exit) intersection.

We assume without loss of generality that all links have a sidewalk on the right side that allows pedestrian movements in both directions, as shown in Figure \ref{fig:ped_net}. At each intersection, pedestrians are allowed to wait at each corner, referred to as pedestrian nodes, to cross the intersection. We use $i_{\alpha}$ to denote a pedestrian node at intersection $i$, where $\alpha$ indicates the relative location of the pedestrian node. Assume each intersection has four corners (i.e., that each intersection has four links). Then, for example, $i_{SE}$ represents the pedestrian node in the southeast of intersection $i$, as shown in Figure \ref{fig:ped_net}. Each sidewalk is represented by a pair of pedestrian nodes. We use $\pazocal{O}^{i_\alpha}$ to indicate the set of adjacent pedestrian nodes of $i_{\alpha}$ at the same intersection, $i$, and $\pazocal{U}^{i_\alpha}$ to denote the set of pedestrian nodes from adjacent intersections of $i$ that connect $i_\alpha$. For the example shown in Figure \ref{fig:ped_net}, $\pazocal{O}^{j_{SW}}=\{j_{NW}, j_{SE}\}$ and $\pazocal{U}^{j_{SW}}=\{i_{SE}, k_{NW}\}$. 

Let $\pazocal{M}^{vi}$ and $\pazocal{M}^{pi}$ indicate the set of vehicle movements and pedestrian movements between two pedestrian nodes at intersection $i$, respectively, and let $\pazocal{S}^i$ be the set of admissible phases, at intersection $i$. An admissible phase at intersection $i$ is represented by an array with a length of $|\pazocal{M}^{vi}|+|\pazocal{M}^{pi}|$, in which each element is a binary variable indicating if the associated movement is served by the phase. Let $S^i(t)$ be the phase that is activated at intersection $i$ at time $t$ and $S^{\mathrm{v},i}_{h,i,j}(t)$ and $S^{\mathrm{p},i}_{i_\alpha, i_\beta}(t)$ be the elements associated with the vehicle movement $(h,i,j)$ and the pedestrian crosswalk movement $(i_\alpha, i_\beta)$ in $S^i(t)$, respectively. For simplicity, we also use the symbols $S^{\mathrm{v}}_{h,i,j}(t)$ and $S^{\mathrm{p}}_{i_\alpha, i_\beta}(t)$ without the superscript for intersection index, which is unique for both vehicle and crosswalk movements, to represent whether the vehicle movement $(h,i,j)$ and the pedestrian crosswalk movement $(i_\alpha, i_\beta)$ are served at time $t$.

\begin{figure}[htb!]
	\includegraphics[width=3.5in]{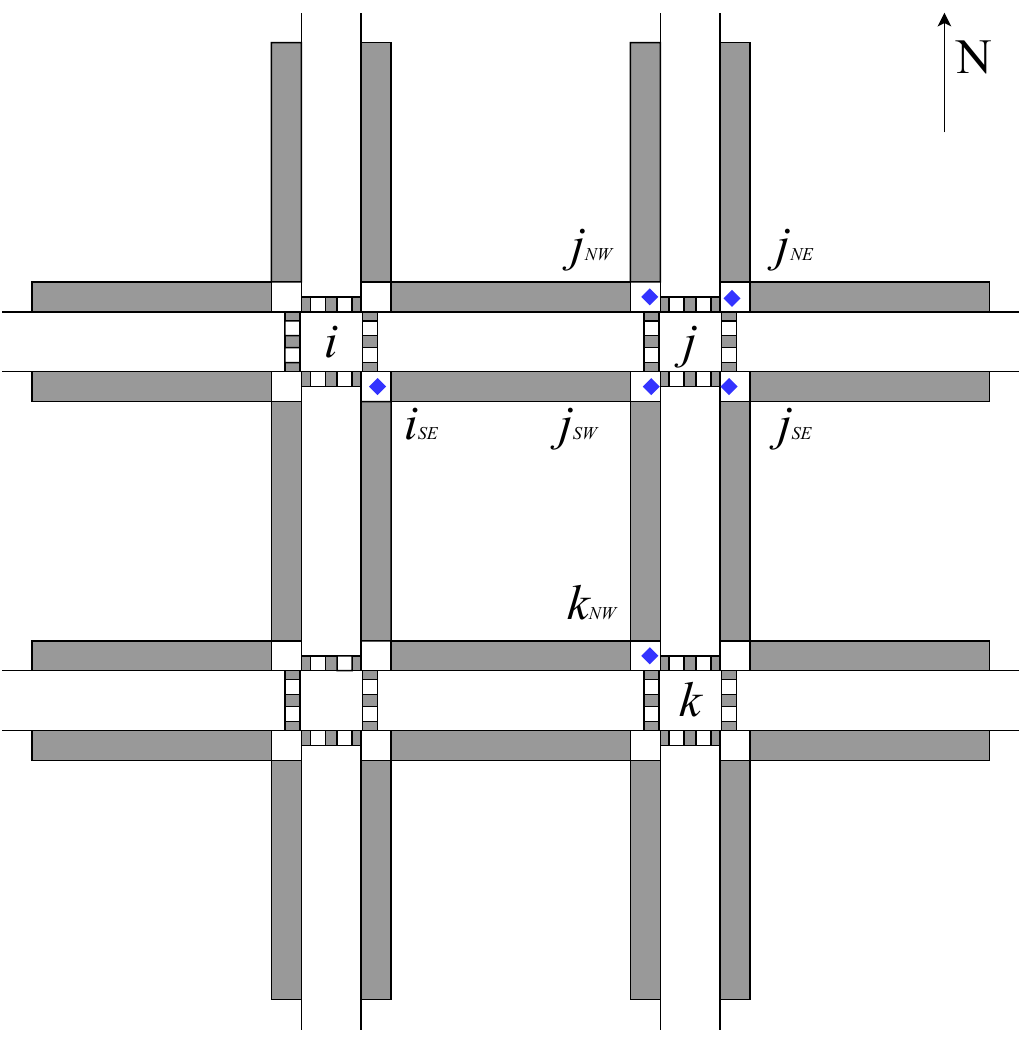}
	\centering
	\caption{Pedestrian network.}
	\label{fig:ped_net}
\end{figure}

\subsection{Vehicle traffic dynamics}
Most MP algorithms in the literature---e.g., \citep{sun2018simulation, varaiya2013max} use a store-and-forward point queue model for the vehicle traffic dynamics in which the maximum occupancy of each link is assumed to be infinite, and this model was observed to be effective in micro-simulations comparing MP to existing signal timings \citep{levin2020max, barman2022performance}. Therefore, this paper use the same model for the vehicle dynamics modeling. Specifically, the evolution of the number of vehicles on link $(h,i,j)$ can be expressed as:

\begin{equation}\label{eq:veh_transition}
\begin{split}
    x^{\mathrm{v}}_{h,i,j}(t+1) = & x^{\mathrm{v}}_{h,i,j}(t)-\min(C^{\mathrm{v}}_{h,i,j}(t)S^{\mathrm{v},i}_{h,i,j}(t),x^{\mathrm{v}}_{h,i,j}(t))\\
    &+\sum_{g\in \pazocal{U}^h}\min(C^{\mathrm{v}}_{g,h,i}(t)S^{\mathrm{v},i}_{g,h,i}(t),x^{\mathrm{v}}_{g,h,i}(t))R^{\mathrm{v}}_{h,i,j}(t+1)\mathbbm{1}_{\pazocal{L}_{\mathrm{in}}}(h,i)\\
    &+d^{\mathrm{v}}_{h,i,j}\mathbbm{1}_{\pazocal{L}_\mathrm{entry}}(h,i)
\end{split}
\end{equation}
where $x^{\mathrm{v}}_{h,i,j}(t)$ is the number of vehicles of movement $(h,i,j)$, i.e., the number of vehicles waiting on link $(h,i)$ to join link $(i,j)$, at time $t$; $C^{\mathrm{v}}_{h,i,j}$ is the stochastic and bounded saturation flow for movement $(h,i,j)$ at time step $t$ with mean value $c^{\mathrm{v}}_{h,i,j}$; $R^{\mathrm{v}}_{g,h,i}$ is the random turning ratio from link $(g,h)$ to link $(h,i)$ with mean value $r^{\mathrm{v}}_{g,h,i}$; $\mathbbm{1}_A(x)$ is the indicator function. The second term in Equation \eqref{eq:veh_transition} is the outflows of movement $(h,i,j)$ during time interval $(t, t+1)$; the third term is the inflows from upstream movements if link $(h,i)$ is an internal link; and, the fourth term is the external demand if link $(h,i)$ is an entry link. Note that the second and third terms in Equation \eqref{eq:veh_transition} imply that the time step size is equal to the free-flow travel time of the link. The long links in the network can be divided into shorter sub-links to satisfy this condition \citep{varaiya2013max}. 

\subsection{Pedestrian dynamics}
Pedestrians are divided into two groups based on their locations: pedestrians waiting at a pedestrian node to cross the intersection---i.e., pedestrian movement between two pedestrian nodes at the same intersection---and pedestrians movement on the sidewalk. For simplicity, the former group is called crosswalk movements while the latter is called sidewalk movements. The following assumptions for pedestrian dynamics are made:
\begin{enumerate}
   \item Pedestrians enter or exit the network from sidewalks;
   \item The maximum occupancy for pedestrian nodes is infinite, so pedestrians cannot be blocked on sidewalks;
   \item Pedestrians on sidewalks always walk at the free-flow speed;
   \item Pedestrians crossing an intersection from a pedestrian node will not go back to the same pedestrian node immediately.
\end{enumerate}

Following these assumptions, the evolution of pedestrians of a crosswalk movement can be expressed as:

\begin{equation}\label{eq:ped_cross_transition}
    \begin{split}
        x^{\mathrm{\mathrm{pc}}}_{i_{\alpha},i_{\beta}}(t+1) = & x^{\mathrm{pc}}_{i_{\alpha},i_{\beta}}(t)-\min(C^{\mathrm{p}}_{i_{\alpha},i_{\beta}}(t)S^{\mathrm{p},i}_{i_{\alpha},i_{\beta}}(t),x^{\mathrm{pc}}_{i_{\alpha},i_{\beta}}(t))\\
        & +\sum_{i_\gamma\in \pazocal{O}^{i_{\alpha}}\setminus{\{i_\beta\}}}\min(C^{\mathrm{p}}_{i_{\gamma},i_{\alpha}}(t)S^{\mathrm{p},i}_{i_{\gamma},i_{\alpha}}(t),x^{\mathrm{pc}}_{i_{\gamma},i_{\alpha}}(t))R^{\mathrm{p}}_{i_{\gamma},i_{\alpha},i_{\beta}}(t+1)\\
        & +\sum_{j_{\eta} \in \pazocal{U}^{i_{\alpha}}}\psi(\mathbf{V}(j_{\eta}, i_{\alpha}))R^{\mathrm{p}}_{j_{\eta},i_{\alpha},i_{\beta}}
    \end{split}
\end{equation}
where $C^{\mathrm{p}}_{i_{\alpha},i_{\beta}}(t)$ is the stochastic saturation flow of crosswalk movement $(i_{\alpha},i_{\beta})$ at time $t$; $R^{\mathrm{p}}_{i_{\gamma},i_{\alpha},i_{\beta}}(t+1)$ is the turning ratio from $(i_{\gamma},i_{\alpha})$ to $(i_{\alpha},i_{\beta})$. The second term is the outflow of crosswalk movement $(i_{\alpha},i_{\beta})$; the third term is the inflow from the adjacent pedestrian node; $\psi$ is a function mapping the distribution of pedestrians on a sidewalk to the number of pedestrians that are going to join movement $(i_{\alpha},i_{\beta})$. Since the walking speed of pedestrians is much slower than the free flow travel speed of vehicles, pedestrians on a sidewalk can arrive the downstream node within one time step only if their distance to the downstream node is shorter than the walking travel distance in one time step. Specifically, $\psi$ can be expressed as:

\begin{equation}
    \psi(\mathbf{V}(j_{\eta}, i_{\alpha}))=\sum_{v\in [1,n^{\mathrm{v}}(j_{\eta}, i_{\alpha})]}\mathbbm{1}_{b(v)\le b_f(v)}
\end{equation}
where $n^{\mathrm{v}}(j_{\eta}, i_{\alpha})$ is the number of pedestrians on $(j_{\eta}, i_{\alpha})$; $b(v)$ is the distance between the $v$th pedestrian to the downstream node; $b_{f(v)}$ is the travel distance of pedestrian $v$ in a time step.

The number of pedestrians of sidewalk movement $(i_{\alpha},j_{\beta})$ can be expressed as,

\begin{equation}\label{eq:ped_side_transition}
\begin{split}
    x^{\mathrm{ps}}_{i_{\alpha},j_{\beta}}(t+1) = & x^{\mathrm{ps}}_{i_{\alpha},j_{\beta}}(t)-\psi(\mathbf{V}(i_{\alpha},j_{\beta}))(t)-q^{\mathrm{out}}_{i_{\alpha},j_{\beta}}(t) + q^{\mathrm{in}}_{i_{\alpha},j_{\beta}}(t) \\
    & +\sum_{i_{\gamma} \in \pazocal{O}^{i_\alpha}}\min(C^{\mathrm{p}}_{i_{\gamma},i_{\alpha}}(t)S^{\mathrm{p}}_{i_{\gamma},i_{\alpha}}(t),x^{\mathrm{pc}}_{i_{\gamma},i_{\alpha}}(t))R^{\mathrm{p}}_{i_{\gamma},i_{\alpha},j_{\beta}}(t+1) \\
    & + \sum_{h_\gamma\in \pazocal{U}^{i_{\alpha}}\setminus{\{j_\beta\}}}\psi(\mathbf{V}(h_{\gamma}, i_{\alpha}))R^{\mathrm{p}}_{h_{\gamma},i_{\alpha},j_{\beta}}
\end{split}
\end{equation}

where $q^{\mathrm{in}}_{i_{\alpha},j_{\beta}}(t)$ and $q^{\mathrm{out}}_{i_{\alpha},j_{\beta}}(t)$ are the number of enter and exit pedestrians on sidewalk $(i_{\alpha},j_{\beta})$ at time $t$, respectively; the fifth term is the inflow from adjacent pedestrian nodes at the same intersection; the sixth term is the inflow from the upstream sidewalk.

\section{Max pressure algorithm incorporating pedestrians}\label{sec:pqmp}
\subsection{Signal control algorithm}
Following the traffic dynamics models, this section derives a step-based MP algorithm that incorporates pedestrian queues into pressure calculation. The proposed algorithm is referred to as PQ-MP in the remainder of this paper. The signal phase at all intersections is updated at a fixed frequency. Like most step-based MP algorithms in the literature, the proposed PQ-MP works in the following four steps.
\begin{enumerate}
    \item At each signal update instant $t$, measure the number of vehicles of all vehicle movements and the number of pedestrians at each pedestrian node;
    
    \item Define the \emph{weight} of a vehicle movement (pedestrian crosswalk movement) as the difference between the number of vehicles (pedestrians) of this movement and the average number of vehicles (pedestrians) over all downstream vehicle movement (pedestrian crosswalk movement) using the turning ratio as the proportion. Specifically, both types of weights can be expressed as Equation \eqref{eq:veh_weight} and Equation \eqref{eq:ped_weight}, respectively. Note that since we assume pedestrians on sidewalks do not incur delay, this group of pedestrians is not considered in the weight definition. In addition, for pedestrian crosswalk movement, we do not consider the downstream sidewalk movement in the weight calculation. In addition, since we assume the pedestrians from node $i_\alpha$ to node $i_\beta$ will not go back to node $i_\alpha$ immediately, there is only one downstream crosswalk movement in Equation \eqref{eq:ped_weight}. 
    
    \begin{equation}\label{eq:veh_weight}
        w^{\mathrm{v}}_{h,i,j}(t) = x^{\mathrm{v}}_{h,i,j}(t)-\sum_{k\in \pazocal{D}^j}x^{\mathrm{v}}_{i,j,k}(t)R_{i,j,k}(t)
    \end{equation}
    \begin{equation}\label{eq:ped_weight}
        w^{\mathrm{p}}_{i_{\alpha},i_{\beta}}(t) = x^{\mathrm{pc}}_{i_{\alpha},i_{\beta}}(t)-x^{\mathrm{pc}}_{i_{\beta},i_{\gamma}}(t)R^{\mathrm{p}}_{i_{\alpha},i_{\beta},i_{\gamma}}
    \end{equation}

    \item Compute the pressure for the phase array $\mathbf{S}^i(t)$ at intersection $i$, as:
    
    \begin{equation}\label{eq:pressure}
        p(\mathbf{S}^i(t))=\sum_{(h,i,j)\in \pazocal{M}^{vi}}w^{\mathrm{v}}_{h,i,j}(t)\tilde{C}^{\mathrm{v}}_{h,i,j}(t)S^{\mathrm{v},i}_{h,i,j}(t)+\lambda\sum_{(i_{\alpha},i_{\beta})\in \pazocal{M}^{pi}}w^{\mathrm{p}}_{i_{\alpha},i_{\beta}}(t)C^{\mathrm{p}}_{i_{\alpha},i_{\beta}}(t)S^{\mathrm{p},i}_{i_{\alpha},i_{\beta}}(t)
    \end{equation}
    
    It is assumed that the right-turn vehicles need to yield to conflicting pedestrian movements served by the same phase. Therefore, the saturation flow for an activated right-turn vehicle movement needs to be adjusted if the conflicting pedestrian crosswalk movements are being served as well. Specifically, the portion of a time step that is required to serve the associated pedestrian cross movement cannot be utilized by the right-turn vehicle movement, so the same proportion needs to be subtracted from the original saturation flow of the right-turn vehicle movements. $\tilde{C}^{\mathrm{v}}_{h,i,j}(t)$ in Equation \eqref{eq:pressure} is the saturation flow of vehicles considering the impact from pedestrians, and it can be expressed as:
    
    \begin{equation}\label{eq:saturation_tilde}
        \tilde{C}^{\mathrm{v}}_{h,i,j}(t)=C^{\mathrm{v}}_{h,i,j}(t)\left(1-\min\left(1, \max_{(i_\alpha, i_\beta)\in \pazocal{M}^{\mathrm{p}}_{h,i,j}}\frac{ x^{\mathrm{pc}}_{i_\alpha, i_\beta}(t)S^{\mathrm{p},i}_{i_\alpha, i_\beta}(t)}{C^{\mathrm{p}}_{i_{\alpha},i_{\beta}}(t)}\right)\right)
    \end{equation}
    where $C^{\mathrm{v}}_{h,i,j}(t)$ is the saturation flow of vehicle movement $(h,i,j)$ without pedestrians, and $\pazocal{M}^{\mathrm{p}}_{h,i,j}$ is the set of pedestrian movements that conflict with vehicle movement $(h,i,j)$. The term of $\frac{ x^{\mathrm{pc}}_{i_\alpha, i_\beta}(t)S^{\mathrm{p},i}_{i_\alpha, i_\beta}(t)}{C^{\mathrm{p}}_{i_{\alpha},i_{\beta}}(t)}$ is the proportion of a time step that is required to clear the pedestrian queue of movement $(i_\alpha, i_\beta)$. Note, since pedestrians are allowed to travel in both directions on a crosswalk, each right-turn vehicle movement has two conflicting crosswalk movements. Therefore, $\min\left(1, \max_{(i_\alpha, i_\beta)\in \pazocal{M}^{\mathrm{p}}_{h,i,j}}\frac{ x^{\mathrm{pc}}_{i_\alpha, i_\beta}(t)S^{\mathrm{p},i}_{i_\alpha, i_\beta}(t)}{C^{\mathrm{p}}_{i_{\alpha},i_{\beta}}(t)}\right)$ is the portion that needs to be subtracted from the saturation flow of the corresponding right-turn movement. 

    In general, the travel speed of vehicles is higher than the walking speed of pedestrians; however, the density of pedestrian stream is higher than the critical density of vehicle flows, and this excess in density usually outweighs the speed shortage. As a result, the saturation flow for pedestrians is significantly higher than that of vehicles. Therefore, without a coefficient for the pedestrian weight, Equation \eqref{eq:pressure} will weigh pedestrian crosswalk movements considerably more over vehicle movements and assign excessive green time to pedestrians. Consequently, the delay incurred by vehicles, especially the left-turn vehicles, could be extremely high. To address this issue, we added an coefficient $0<\lambda<1$ for the second term of Equation \eqref{eq:pressure}. 
    
    \item At each intersection, the phase with the maximum pressure is activated for the next time step,

    \begin{equation}\label{eq:opt}
        S^*_i(t)=\argmax_{S^{i}(t)\in \pazocal{S}^i}p(S^{i}(t))
    \end{equation}
\end{enumerate}

\subsection{Maximum stability property}
One desirable property of MP-based algorithms is the maximum stability, which suggests that  they can accommodate a demand if this demand can be accommodated by any feasible control strategy. This section proves this property holds for the PQ-MP.

\subsubsection{Stable region}
Let $\mathbf{d}^{\mathrm{v}}$ indicate the vector of (average) vehicular demand from entry links, $\mathbf{q}^{\mathrm{in}}$ indicate the vector of (average) pedestrian generation rate from all sidewalks, and $\mathbf{q}^{\mathrm{out}}$ indicate the vector of (average) pedestrian exit rate from all sidewalks. In addition, let $\mathbf{f}^{\mathrm{v}}$ indicate the vector of (average) vehicular demand on all internal links, $\mathbf{f}^{\mathrm{ps}}$ indicate the vector of (average) pedestrian demand on sidewalks, and $\mathbf{f}^{\mathrm{pc}}$ indicate the vector of (average) pedestrian demand for crosswalks. $\mathbf{f}^{\mathrm{v}}$, $\mathbf{f}^{\mathrm{ps}}$ and $\mathbf{f}^{\mathrm{pc}}$ can be uniquely determined by $\mathbf{d}^{\mathrm{v}}$, $\mathbf{q}^{\mathrm{in}}$, $\mathbf{q}^{\mathrm{out}}$ and turning ratios for both vehicles and pedestrians \citep{hao2018model}. It is easy to obtain the expressions in Equations \eqref{eq:fv_dv}--\eqref{eq:fps_dp}.

\begin{subequations}\label{eq:fv_dv}
	\begin{empheq}[left={f^{\mathrm{v}}_{l,m,n}=\empheqlbrace\,}]{alignat=2}
	& d^{\mathrm{v}}_{l,m,n}, && \quad \text{if $(l,m)\in \pazocal{L_{\mathrm{enter}}}$}\\
	&\sum_{k\in \pazocal{U}^l}f^{\mathrm{v}}_{k,l,m}r^{\mathrm{v}}_{l,m,n} && \quad \text{otherwise}
	\end{empheq}
\end{subequations}

\begin{equation}\label{eq:fpc_dp}
    f^{\mathrm{pc}}_{i_\alpha, i_\beta}=\sum_{i_\gamma\in \pazocal{O}^{i_{\alpha}}\setminus{\{i_\beta\}}}f^{\mathrm{pc}}_{i_\gamma, i_\alpha}r^{\mathrm{p}}_{i_{\gamma},i_{\alpha},i_{\beta}}+\sum_{j_{\eta} \in \pazocal{U}^{i_{\alpha}}}f^{\mathrm{ps}}_{j_\eta, i_\alpha}r^{\mathrm{p}}_{j_{\eta},i_{\alpha},i_{\beta}}, \quad \forall (i_\alpha, i_\beta)
\end{equation}

\begin{equation}\label{eq:fps_dp}
    f^{\mathrm{ps}}_{i_\alpha, j_\beta}=q^{\mathrm{in}}_{i_{\alpha},j_{\beta}} - q^{\mathrm{out}}_{i_{\alpha},j_{\beta}}+\sum_{i_{\gamma} \in \pazocal{O}^{i_\alpha}}f^{\mathrm{pc}}_{i_\gamma, i_\alpha}r^{\mathrm{p}}_{i_{\gamma},i_{\alpha},j_{\beta}}+\sum_{h_\gamma\in \pazocal{U}^{i_{\alpha}}\setminus{\{j_\beta\}}}f^{\mathrm{ps}}_{h_{\gamma},i_{\alpha}}r^{\mathrm{p}}_{h_{\gamma},i_{\alpha},j_{\beta}}, \quad \forall (i_\alpha, j_\beta)
\end{equation}

The convex hull of all admissible phases at intersection $i$ can be expressed as:

\begin{equation}\label{eq:ch_phase}
    co(\pazocal{S}^i)=\{\sum_{\mathbf{S}^{i,e}\in \pazocal{S}^i}\pi_e\mathbf{S}^{i,e}|\pi_e\ge 0, \sum_{e}\pi_e=1\}
\end{equation}

\begin{definition}\label{def:feasible}
    A demand $(\mathbf{d}^{\mathrm{v}}, \mathbf{p}^{\mathrm{in}})$ is feasible if there exists a control vector $\boldsymbol{\Sigma}\in co(\pazocal{S})$ such that

    \begin{equation}\label{eq:stable_region_sw}
        f^{\mathrm{ps}}_{i_{\alpha}, j_{\beta}}\le c^{\mathrm{ps}}_{i_{\alpha}, j_{\beta}},\quad \forall (i_{\alpha}, j_{\beta})
    \end{equation}
    \begin{equation}\label{eq:stable_region_veh}
        f^{\mathrm{v}}_{l,m,n}\le \tilde{c}^{\mathrm{v}}_{l,m,n}\Sigma^\mathrm{v}_{l,m,n}, \quad \forall (l,m,n)
    \end{equation}
    \begin{equation}\label{eq:stable_region_cw}
        f^{\mathrm{pc}}_{i_{\alpha},i_{\beta}}\le c^{\mathrm{p}}_{i_{\alpha},i_{\beta}}\Sigma^\mathrm{p}_{i_{\alpha},i_{\beta}}, \quad \forall (i,\alpha,\beta)
    \end{equation}
    where $\tilde{c}^{\mathrm{v}}_{l,m,n}$ and $c^{\mathrm{p}}_{i_{\alpha},i_{\beta}}$ are the mean of the saturation flows for vehicle movement and pedestrian crosswalk movement, respectively. 
\end{definition}

It is easy to prove that a control sequence $\{\mathbf{S}(t), t=0, T, 2T, 3T,...\}$ is admissible if and only if

\begin{equation}
    \liminf\limits_{N\rightarrow \infty}\frac{1}{N}\sum_{n=0}^{N}\mathbf{S}^i(nT) \in co(\pazocal{S}^i), \quad \forall i \in \pazocal{N}
\end{equation}

Therefore, Definition \ref{def:feasible} indicates that a control sequence can accommodate a demand, or equivalently, a demand can be stabilized by the control sequence, if the average service rate is not lower than the average demand for all movements in the network. Equation \eqref{eq:stable_region_veh} and Equation \eqref{eq:stable_region_cw} indicate the vehicle and pedestrian demand for crosswalks are lower than the corresponding average service rates. Note that if Equation \ref{eq:stable_region_cw} holds, the third and fourth terms in Equation \ref{eq:stable_region_sw} are independent of the signal control strategy and uniquely determined by $\mathbf{q}^{\mathrm{in}}$, $\mathbf{q}^{\mathrm{out}}$ and $\mathbf{r}^{\mathrm{p}}$. Then, Equation \eqref{eq:stable_region_sw} indicates that the sum of pedestrian generating rate and the external demand from adjacent sidewalks and cross movements subtracted by the exit rate on sidewalks should be lower than the corresponding saturation flow to make the pedestrian demand feasible. This relationship is independent of signal timings under the assumptions made in the previous section. The stable region of demand, denoted by $(\pazocal{D}^{\mathrm{v}}, \pazocal{Q}_{\mathrm{in}})$, is defined as the set of demand that satisfies Equations \eqref{eq:stable_region_sw}-\eqref{eq:stable_region_cw}. Let $\pazocal{D}^0$ and $\pazocal{Q}_{\mathrm{in}}^0$ denote the interior of $\pazocal{D}^{\mathrm{v}}$ and $\pazocal{Q}_{\mathrm{in}}$, respectively. 

\subsubsection{Maximum stability}
\begin{definition}
    A signal control policy $\pazocal{S}$ stabilizes the queue process in the mean if for some $M^{\mathrm{v}}<\infty$ and $M^{\mathrm{p}}<\infty$,
    
    \begin{equation}\label{eq:def_ms_v}
        \frac{1}{T}\sum_{t=1}^T\sum_{l,m,n}E(x^{\mathrm{v}}_{l,m,n}(t))\le M^{\mathrm{v}}, \quad T=1,2,3...
    \end{equation}
    \begin{equation}\label{eq:def_ms_p}
        \frac{1}{T}\sum_{t=1}^T\sum_{i,\alpha,\beta}E(x^{\mathrm{pc}}_{i_{\alpha},i_{\beta}}(t))\le M^{\mathrm{p}}, \quad T=1,2,3...
    \end{equation}
\end{definition}
As mentioned before, the average number of pedestrians on sidewalks is stabilized regardless of the signal control strategy if the demand is in the stable region. The combination of Equation \eqref{eq:def_ms_v} and Equation \eqref{eq:def_ms_p} is equivalent to 

\begin{equation}\label{eq:def_ms}
    \frac{1}{T}\sum_{t=1}^T\left[\sum_{l,m,n}E(x^{\mathrm{v}}_{l,m,n}(t))+\sum_{i,\alpha,\beta}E(x^{\mathrm{pc}}_{i_{\alpha},i_{\beta}}(t))\right]\le M, \quad T=1,2,3...
\end{equation}
where $M<\infty$.

\begin{theorem}[Maximum stability]\label{theorem:ms}
    The proposed PQ-MP algorithm Equations \eqref{eq:veh_weight}--\eqref{eq:opt} stabilizes the queue process if $(\mathbf{d}^{\mathrm{v}}, \mathbf{p}^{\mathrm{in}})\in (\pazocal{D}^0, \pazocal{Q}_{\mathrm{in}}^0)$ and if the distribution of the adjusted saturation flow, $\tilde{C}^{\mathrm{v}}_{h,i,j}(t)$ in Equation \eqref{eq:pressure}, is fixed for all right-turn vehicle movements.
\end{theorem}

\begin{proof}
The first step is to prove that there exist $\epsilon_1>0$ and $\epsilon_2>0$ such that the Lyapunov functions $|\mathbf{X}^\mathrm{v}(t)|^2$ and $|\sqrt{\lambda}\mathbf{X}^{\mathrm{pc}}(t)|^2$ under the control of PQ-MP satisfy the following inequality:

\begin{equation}\label{eq:key}
    E\{|\mathbf{X}^\mathrm{v}(t+1)|^2-|\mathbf{X}^\mathrm{v}(t)|^2+|\sqrt{\lambda}\mathbf{X}^{\mathrm{pc}}(t+1)|^2-|\sqrt{\lambda}\mathbf{X}^{\mathrm{pc}}(t)|^2|\mathbf{X}^\mathrm{v}(t),\mathbf{X}^{\mathrm{pc}}(t)\}\le k-\epsilon_1|\mathbf{X}^\mathrm{v}(t)|-\epsilon_2|\mathbf{X}^{\mathrm{pc}}(t)|
\end{equation}
where $|\mathbf{X}|\equiv\sum_{x_i \in \mathbf{X}}x_i$,  $|\mathbf{X}|^2\equiv\sum_{x_i \in \mathbf{X}}(x_i)^2$, and $|\sqrt{\lambda}\mathbf{X}|^2\equiv\lambda\sum_{x_i \in \mathbf{X}}(x_i)^2$.

Using Eqs. \eqref{eq:veh_transition} and \eqref{eq:saturation_tilde}, we obtain:

\begin{equation}\label{eq:diff}
\begin{split}
    \delta^{\mathrm{v}}_{l,m,n}(t)=&{x^{\mathrm{v}}_{l,m,n}(t+1)}-x^{\mathrm{v}}_{l,m,n}(t)\\
    =&-\min(\tilde{C}^{\mathrm{v}}_{l,m,n}(t)S^{\mathrm{v}}_{l,m,n}(t),x^{\mathrm{v}}_{l,m,n}(t))\\
    &+\sum_k\min(\tilde{C}^{\mathrm{v}}_{k,l,m}(t)S^{\mathrm{v}}_{k,l,m}(t),x^{\mathrm{v}}_{k,l,m}(t))R^{\mathrm{v}}_{l,m,n}(t+1)\mathbbm{1}_{\pazocal{L}_{\mathrm{in}}}(l,m)\\
    &+d^{\mathrm{v}}_{l,m,n}\mathbbm{1}_{\pazocal{L}_\mathrm{enter}}(l,m)
\end{split}
\end{equation}

Equation \eqref{eq:diff} leads to

\begin{equation}\label{eq:diff_veh}
 |\mathbf{X}^{\mathrm{v}}(t+1)|^2-|\mathbf{X}^{\mathrm{v}}(t)|^2=2\mathbf{X}^{\mathrm{v}}(t)^T\mathbf{\delta}^{\mathrm{v}}(t)+|\mathbf{\delta}^{\mathrm{v}}(t)|^2=2\omega^{\mathrm{v}}+\chi^{\mathrm{v}}
\end{equation}
where

\begin{equation}\label{eq:beta_v}
\begin{split}
    \chi^{\mathrm{v}}=&\sum_{l,m,n}(x^{\mathrm{v}}_{l,m,n}(t+1)-x^{\mathrm{v}}_{l,m,n}(t))^2\\
    & = \sum_{l,m,n}\{-\min(\tilde{C}^{\mathrm{v}}_{l,m,n}(t)S^{\mathrm{v}}_{l,m,n}(t),x^{\mathrm{v}}_{l,m,n}(t))\\
    &+\sum_k\min(\tilde{C}^{\mathrm{v}}_{k,l,m}(t)S^{\mathrm{v}}_{k,l,m}(t),x^{\mathrm{v}}_{k,l,m}(t))R^{\mathrm{v}}_{l,m,n}(t+1)\mathbbm{1}_{\pazocal{L}_{\mathrm{in}}}(l,m)\\
    &+d^{\mathrm{v}}_{l,m,n}\mathbbm{1}_{\pazocal{L}_\mathrm{enter}}(l,m)\}^2
\end{split} 
\end{equation}
and
\begin{subequations}
\begin{align}
        \omega^{\mathrm{v}}=&\mathbf{X}^{\mathrm{v}}(t)^T\mathbf{\delta}^{\mathrm{v}}(t)\\
        =&-\sum_{l,m,n}x^{\mathrm{v}}_{l,m,n}(t)\min(\tilde{C}^{\mathrm{v}}_{l,m,n}(t)S^{\mathrm{v}}_{l,m,n}(t),x^{\mathrm{v}}_{l,m,n}(t))\\
    &+\sum_{l,m,n}x^{\mathrm{v}}_{l,m,n}(t)\sum_k\min(\tilde{C}^{\mathrm{v}}_{k,l,m}(t)S^{\mathrm{v}}_{k,l,m}(t),x^{\mathrm{v}}_{k,l,m}(t))R^{\mathrm{v}}_{l,m,n}(t+1)\mathbbm{1}_{\pazocal{L}_{\mathrm{in}}}(l,m)\\
    &+\sum_{l,m,n}x^{\mathrm{v}}_{l,m,n}(t)d^{\mathrm{v}}_{l,m,n}\mathbbm{1}_{\pazocal{L}_\mathrm{enter}}(l,m)\\
    =&-\sum_{l,m,n}x^{\mathrm{v}}_{l,m,n}(t)\min(\tilde{C}^{\mathrm{v}}_{l,m,n}(t)S^{\mathrm{v}}_{l,m,n}(t),x^{\mathrm{v}}_{l,m,n}(t))\\
    &+\sum_{m,n,k}x^{\mathrm{v}}_{m,n,k}(t)\sum_l\min(\tilde{C}^{\mathrm{v}}_{l,m,n}(t)S^{\mathrm{v}}_{l,m,n}(t),x^{\mathrm{v}}_{l,m,n}(t))R^{\mathrm{v}}_{m,n,k}(t+1)\mathbbm{1}_{\pazocal{L}_{\mathrm{in}}}(m,n)\\
    &+\sum_{l,m,n}x^{\mathrm{v}}_{l,m,n}(t)d^{\mathrm{v}}_{l,m,n}\mathbbm{1}_{\pazocal{L}_\mathrm{enter}}(l,m)\\
    =&-\sum_{l,m,n}\min(\tilde{C}^{\mathrm{v}}_{l,m,n}(t)S^{\mathrm{v}}_{l,m,n}(t),x^{\mathrm{v}}_{l,m,n}(t))\left(x^{\mathrm{v}}_{l,m,n}(t)-\sum_{k\in \pazocal{D}^n}x^{\mathrm{v}}_{m,n,k}(t)R^{\mathrm{v}}_{m,n,k}(t+1)\mathbbm{1}_{\pazocal{L}_{\mathrm{in}}}(m,n)\right)\\
    &+\sum_{l,m,n}x^{\mathrm{v}}_{l,m,n}(t)d^{\mathrm{v}}_{l,m,n}\mathbbm{1}_{\pazocal{L}_\mathrm{enter}}(l,m)
\end{align}
\end{subequations}

Note that the expression of $\omega^{\mathrm{v}}$ has the same expression in \citep{varaiya2013max}. Using Equation \eqref{eq:fv_dv} and the same manner as in \citep{varaiya2013max}, it is easy to obtain

\begin{equation}\label{eq:omega_v}
    E[\omega^{\mathrm{v}}|\mathbf{X^{\mathrm{v}}(t)}]=\omega^{\mathrm{v1}}+\omega^{\mathrm{v2}}
\end{equation}
where
\begin{equation}\label{eq:alpha1_v}
    \omega^{\mathrm{v1}}=\sum_{l,m,n}\left[f^{\mathrm{v}}_{l,m,n}-\tilde{c}^{\mathrm{v}}_{l,m,n}(t)S^{\mathrm{v}}_{l,m,n}(t)\right]w^{\mathrm{v}}_{l,m,n}(t+1)
\end{equation}
\begin{equation}
    \omega^{\mathrm{v2}}=\sum_{l,m,n}\left[\tilde{c}^{\mathrm{v}}_{l,m,n}(t)-E[\min(\tilde{C}^{\mathrm{v}}_{l,m,n}(t), x^{\mathrm{v}}_{l,m,n}(t))|\mathbf{X^{\mathrm{v}}(t)}]\right]S^{\mathrm{v}}_{l,m,n}(t)w^{\mathrm{v}}_{l,m,n}(t+1)
\end{equation}
where $\tilde{c}^{\mathrm{v}}_{l,m,n}(t) = E[\tilde{C}^{\mathrm{v}}_{l,m,n}(t)|\mathbf{X^{\mathrm{v}}(t)}]$.
\begin{lemma}\label{prop:beta_alpha2}
    $\omega^{\mathrm{v2}}$ and $E[\chi^{\mathrm{v}}|\mathbf{X^{\mathrm{v}}(t)}]$ are always bounded by a positive number regardless of the control policy. 
\end{lemma}
The proof for Lemma \ref{prop:beta_alpha2} can be found in \citep{varaiya2013max}, and it is omitted for simplicity.

Next, for pedestrian crosswalk movements, according to Equation \eqref{eq:ped_cross_transition}, we obtain

\begin{equation}\label{eq:ped_delta}
\begin{split}
    \delta^{\mathrm{pc}}(i_{\alpha},i_{\beta})(t) = & x^{\mathrm{pc}}_{i_{\alpha},i_{\beta}}(t+1)-x^{\mathrm{pc}}_{i_{\alpha},i_{\beta}}(t)\\
    = & -\min(C^{\mathrm{p}}_{i_{\alpha},i_{\beta}}(t)S^{\mathrm{p}}_{i_{\alpha},i_{\beta}}(t),x^{\mathrm{pc}}_{i_{\alpha},i_{\beta}}(t))+\min(C^{\mathrm{p}}_{i_{\gamma},i_{\alpha}}(t)S^{\mathrm{p}}_{i_{\gamma},i_{\alpha}}(t),x^{\mathrm{pc}}_{i_{\gamma},i_{\alpha}}(t))R^{\mathrm{p}}_{i_{\gamma},i_{\alpha},i_{\beta}}\\
    & +\sum_{j_{\eta} \in \pazocal{U}^{i_{\alpha}}}\psi(\mathbf{V}(j_{\eta}, i_{\alpha}))R^{\mathrm{p}}_{j_{\eta},i_{\alpha},i_{\beta}}
\end{split}
\end{equation}

Then,
\begin{equation}\label{eq:diff_ped}
    |\sqrt{\lambda}\mathbf{X}^{\mathrm{pc}}(t+1)|^2-|\sqrt{\lambda}\mathbf{X}^{\mathrm{pc}}(t)|^2=2\lambda\mathbf{X}^{\mathrm{pc}}(t)^T\mathbf{\delta}^{\mathrm{pc}}(t)+\lambda|\mathbf{\delta}^{\mathrm{pc}}(t)|^2=2\lambda\omega^{\mathrm{pc}}+\lambda\chi^{\mathrm{pc}}
\end{equation}
where

\begin{equation}
    \chi^{\mathrm{pc}} = |\mathbf{\delta}^{\mathrm{pc}}(t)|^2
\end{equation}
and

\begin{subequations}\label{eq:alpha_ped}
\begin{align}
    \omega^{\mathrm{pc}}=&\mathbf{X}^{\mathrm{pc}}(t)^T\mathbf{\delta}^{\mathrm{pc}}(t)\\
    =& -\sum_{i,\alpha,\beta}x^{\mathrm{pc}}_{i_{\alpha}, i_{\beta}}(t)\min\left(C^{\mathrm{p}}_{i_{\alpha},i_{\beta}}(t)S^{\mathrm{p}}_{i_{\alpha},i_{\beta}}(t),x^{\mathrm{pc}}_{i_{\alpha},i_{\beta}}(t)\right)\\
    &+ \sum_{i,\alpha,\beta}x^{\mathrm{pc}}_{i_{\alpha}, i_{\beta}}(t)\min(C^{\mathrm{p}}_{i_{\gamma},i_{\alpha}}(t)S^{\mathrm{p}}_{i_{\gamma},i_{\alpha}}(t),x^{\mathrm{pc}}_{i_{\gamma},i_{\alpha}}(t))R^{\mathrm{p}}_{i_{\gamma},i_{\alpha},i_{\beta}}\\
    &+ \sum_{i,\alpha,\beta}x^{\mathrm{pc}}_{i_{\alpha}, i_{\beta}}(t)\left(\sum_{j_{\eta} \in \mathbf{U}_{i_{\alpha}}}\psi(\mathbf{V}(j_{\eta}, i_{\alpha}))R^{\mathrm{p}}_{j_{\eta},i_{\alpha},i_{\beta}}\right)\\
    =& \omega^{\mathrm{pc1}}_{i_{\alpha}, i_{\beta}}(t)+\omega^{\mathrm{pc2}}_{i_{\alpha}, i_{\beta}}(t)+\omega^{\mathrm{pc3}}_{i_{\alpha}, i_{\beta}}(t)
\end{align}
\end{subequations}
\begin{lemma}
    $\lambda \chi^{\mathrm{pc}}$ is upper bounded by a positive constant.
\end{lemma}
\begin{proof}
    Since the random saturation flow for pedestrians is bounded, the first two terms in Equation \eqref{eq:ped_delta} are bounded. In addition, according to Equation \eqref{eq:stable_region_sw}, the number of pedestrians on sidewalks is bounded. Therefore, $\psi(\mathbf{V}(j_{\eta}, i_{\alpha}))$ in Equation \eqref{eq:ped_delta} is also bounded. Consequently, $\lambda\chi^{\mathrm{pc}}$ is upper bounded.
\end{proof}

Next, we focus on $\omega^{\mathrm{pc}}$. Changing the sequence of the subscripts for $\omega^{\mathrm{pc2}}_{i_{\alpha}, i_{\beta}}(t)$ obtains,
\begin{equation}
    \omega^{\mathrm{pc2}}_{i_{\alpha}, i_{\beta}}(t)=\sum_{i,\alpha,\beta}x^{\mathrm{pc}}_{i_{\beta}, i_{\gamma}}(t)\min(C^{\mathrm{p}}_{i_{\alpha},i_{\beta}}(t)S^{\mathrm{p}}_{i_{\alpha},i_{\beta}}(t),x^{\mathrm{pc}}_{i_{\alpha},i_{\beta}}(t))R^{\mathrm{p}}_{i_{\alpha},i_{\beta},i_{\gamma}}
\end{equation}

Therefore, according to Equation \eqref{eq:ped_weight}, we obtain

\begin{equation}\label{eq:alpha12}
\begin{split}
    &\omega^{\mathrm{pc1}}_{i_{\alpha}, i_{\beta}}(t)+\omega^{\mathrm{pc2}}_{i_{\alpha}, i_{\beta}}(t)\\
    =& -\sum_{i,\alpha,\beta}\min(C^{\mathrm{p}}_{i_{\alpha},i_{\beta}}(t)S^{\mathrm{p}}_{i_{\alpha},i_{\beta}}(t),x^{\mathrm{pc}}_{i_{\alpha},i_{\beta}}(t))\left(x^{\mathrm{pc}}_{i_{\alpha}, i_{\beta}}(t)-x^{\mathrm{pc}}_{i_{\beta}, i_{\gamma}}(t)R^{\mathrm{p}}_{i_{\alpha},i_{\beta},i_{\gamma}}\right)\\
    =& -\sum_{i,\alpha,\beta}\min(C^{\mathrm{p}}_{i_{\alpha},i_{\beta}}(t)S^{\mathrm{p}}_{i_{\alpha},i_{\beta}}(t),x^{\mathrm{pc}}_{i_{\alpha},i_{\beta}}(t))w^{\mathrm{p}}_{i_{\alpha}, i_{\beta}}(t+1)
\end{split}
\end{equation}

The term in the parenthesis in $\omega^{\mathrm{pc3}}_{i_{\alpha}, i_{\beta}}(t)$ in Equation \eqref{eq:alpha_ped} is the average pedestrian demand from upstream sidewalks to a crosswalk, which is equivalent to the total pedestrian demand for the crosswalk movement subtracted by the pedestrian demand from upstream crosswalks. Therefore,

\begin{subequations}\label{eq:alpha3}
    \begin{align}
        E[\omega^{\mathrm{pc3}}_{i_{\alpha}, i_{\beta}}(t)|\mathbf{X}^{\mathrm{pc}}(t)]=&\sum_{i,\alpha,\beta}x^{\mathrm{pc}}_{i_{\alpha},i_{\beta}}(t)\left( f^{\mathrm{pc}}_{i_{\alpha}, i_{\beta}}(t)-f^{\mathrm{pc}}_{i_{\gamma}, i_{\alpha}}(t)R^{\mathrm{p}}_{i_{\gamma}, i_{\alpha},i_{\beta}}(t) \right) \\
        & \text{(For a crosswalk movement, only one upstream crosswalk movement exists)} \nonumber\\
        =&\sum_{i,\alpha,\beta}\left[x^{\mathrm{pc}}_{i_{\alpha},i_{\beta}}(t) f^{\mathrm{pc}}_{i_{\alpha}, i_{\beta}}(t)-x^{\mathrm{pc}}_{i_{\alpha},i_{\beta}}(t)f^{\mathrm{pc}}_{i_{\gamma}, i_{\alpha}}(t)R^{\mathrm{p}}_{i_{\gamma}, i_{\alpha},i_{\beta}}(t) \right]\\
        =&\sum_{i,\alpha,\beta}\left[x^{\mathrm{pc}}_{i_{\alpha},i_{\beta}}(t) f^{\mathrm{pc}}_{i_{\alpha}, i_{\beta}}(t)-x^{\mathrm{pc}}_{i_{\beta},i_{\gamma}}(t)f^{\mathrm{pc}}_{i_{\alpha}, i_{\beta}}(t)R^{\mathrm{p}}_{i_{\alpha},i_{\beta}, i_{\gamma}}(t)\right]\\
        =&\sum_{i,\alpha,\beta}f^{\mathrm{pc}}_{i_{\alpha}, i_{\beta}}w^{\mathrm{p}}_{i_{\alpha}, i_{\beta}}(t+1)
    \end{align}
\end{subequations}

Therefore, by combining Equation \eqref{eq:alpha12} and Equation \eqref{eq:alpha3}, we have

\begin{subequations}\label{eq:alpha^p_c4}
\begin{align}
    E\left[\omega^{\mathrm{pc}}|\mathbf{X}^{\mathrm{pc}}(t)\right]=&E\left[\sum_{i,\alpha,\beta}\left[f^{\mathrm{pc}}_{i_{\alpha}, i_{\beta}}-\min(C^{\mathrm{p}}_{i_{\alpha},i_{\beta}}(t)S^{\mathrm{p}}_{i_{\alpha},i_{\beta}}(t),x^{\mathrm{pc}}_{i_{\alpha},i_{\beta}}(t))\right]w^{\mathrm{p}}_{i_{\alpha}, i_{\beta}}(t+1)|\mathbf{X}^{\mathrm{pc}}(t)\right]\\
    =&\sum_{i,\alpha,\beta}\left[f^{\mathrm{pc}}_{i_{\alpha}, i_{\beta}}-c^{\mathrm{p}}_{i_{\alpha},i_{\beta}}(t)S^{\mathrm{p}}_{i_{\alpha},i_{\beta}}(t)\right]w^{\mathrm{p}}_{i_{\alpha}, i_{\beta}}(t+1)\\
    &+\left[\sum_{i,\alpha,\beta}\left[c^{\mathrm{p}}_{i_{\alpha},i_{\beta}}(t)S^{\mathrm{p}}_{i_{\alpha},i_{\beta}}(t)-\min(C^{\mathrm{p}}_{i_{\alpha},i_{\beta}}(t)S^{\mathrm{p}}_{i_{\alpha},i_{\beta}}(t),x^{\mathrm{pc}}_{i_{\alpha},i_{\beta}}(t))\right]w^{\mathrm{p}}_{i_{\alpha}, i_{\beta}}(t+1)|\mathbf{X}^{\mathrm{pc}}(t)\right]\\
    =& \omega^{\mathrm{pc4}}+\omega^{\mathrm{pc5}}
\end{align}
\end{subequations}

Again, using the same manner as in Lemma \ref{prop:beta_alpha2} can easily prove that $\omega^{\mathrm{pc5}}$ is upper bounded by a positive number. 

Combining Equations \eqref{eq:diff_veh}, \eqref{eq:omega_v}, \eqref{eq:diff_ped} and \eqref{eq:alpha^p_c4} obtains

\begin{equation}\label{eq:transform}
\begin{split}
    &E\{|\mathbf{X}^\mathrm{v}(t+1)|^2-|\mathbf{X}^\mathrm{v}(t)|^2+|\sqrt{\lambda}\mathbf{X}^{\mathrm{pc}}(t+1)|^2-|\sqrt{\lambda}\mathbf{X}^{\mathrm{pc}}(t)|^2|\mathbf{X}^\mathrm{v}(t),\mathbf{X}^{\mathrm{pc}}(t)\}\\
    =& 2\omega^{\mathrm{v1}}+2\omega^{\mathrm{v2}}+E[\chi^{\mathrm{v}}|\mathbf{X^{\mathrm{v}}(t)}]+2\lambda\omega^{\mathrm{pc4}}+2\lambda\omega^{\mathrm{pc5}}+\lambda E[\chi^{\mathrm{pc}}|\mathbf{X^{\mathrm{pc}}(t)}]\\
    =&2\left(\omega^{\mathrm{v1}}+\lambda\omega^{\mathrm{pc4}}\right)+\left(2\omega^{\mathrm{v2}}+E[\chi^{\mathrm{v}}|\mathbf{X^{\mathrm{v}}(t)}]+2\lambda\omega^{\mathrm{pc5}}+\lambda E[\chi^{\mathrm{pc}}|\mathbf{X^{\mathrm{pc}}(t)}]\right)
\end{split}
\end{equation}

We already proved that the term in the second parenthesis in Equation \eqref{eq:transform} is bounded, so we next focus on $\omega^{\mathrm{v1}}+\lambda\omega^{\mathrm{pc4}}$. Let $\mathbf{S}^{*}(t)$ be the optimal control array from PQ-MP. Since $\mathbf{S}^{*}(t)$ maximizes $\sum_{(h,i,j)\in \pazocal{M}^{vi}}w^{\mathrm{v}}_{h,i,j}\tilde{C}^{\mathrm{v}}_{h,i,j}S^{ij}_{h,i,j}+\lambda\sum_{(i_{\alpha},i_{\beta})\in \pazocal{M}^{pi}}w^{\mathrm{p}}_{i_{\alpha},i_{\beta}}C^{\mathrm{p}}_{i_{\alpha},i_{\beta}}S^{ij}_{i_{\alpha},i_{\beta}}$ which is a linear function of $\mathbf{S}(t)$, based on the fundamental theorem of linear programming, $\mathbf{S}^{*}(t)$ also maximizes this function over the entire convex hull. By combining Eqs. \eqref{eq:alpha1_v} and \eqref{eq:alpha^p_c4}, for any control array $\mathbf{S}\in co(\pazocal{S})$, we have

\begin{equation}\label{eq:alphas}
\begin{split}
    &\omega^{\mathrm{v1}}+\lambda\omega^{\mathrm{pc4}}\\
    =&\sum_{l,m,n}\left[f^{\mathrm{v}}_{l,m,n}-\tilde{c}^{\mathrm{v}}_{l,m,n}(t)S^{\mathrm{v}*}_{l,m,n}(t)\right]w^{\mathrm{v}}_{l,m,n}(t+1)+\sum_{i,\alpha,\beta}\left[f^{\mathrm{pc}}_{i_{\alpha}, i_{\beta}}-c^{\mathrm{p}}_{i_{\alpha},i_{\beta}}S^{\mathrm{p}*}_{i_{\alpha},i_{\beta}}(t)\right]\lambda w^{\mathrm{p}}_{i_{\alpha}, i_{\beta}}(t+1)\\
    \\
    \le&\sum_{l,m,n}\left[f^{\mathrm{v}}_{l,m,n}-\tilde{c}^{\mathrm{v}}_{l,m,n}(t)S^{\mathrm{v}}_{l,m,n}(t)\right]w^{\mathrm{v}}_{l,m,n}(t+1)+\sum_{i,\alpha,\beta}\left[f^{\mathrm{pc}}_{i_{\alpha}, i_{\beta}}-c^{\mathrm{p}}_{i_{\alpha},i_{\beta}}S^{\mathrm{p}}_{i_{\alpha},i_{\beta}}(t)\right]\lambda w^{\mathrm{p}}_{i_{\alpha}, i_{\beta}}(t+1)
\end{split}
\end{equation}
% \\
%     \le& -\epsilon_1|X^{\mathrm{v}}(t)|-\epsilon_2|x^{\mathrm{pc}}(t)|

Since $(\mathbf{d}^{\mathrm{v}},\mathbf{q}_{\mathrm{in}})$ in $(\pazocal{D}^0, \pazocal{Q}_{\mathrm{in}}^0)$, there exist $\epsilon_1>0, \epsilon_2>0$ and a control array $\mathbf{S}\in co(\pazocal{S})$ such that 

\begin{subequations}\label{eq:ineq1}
\begin{align}
\tilde{c}^{\mathrm{v}}_{l,m,n}(t)S^{\mathrm{v}}_{l,m,n} \ge f^{\mathrm{v}}_{l,m,n}+\epsilon_1 & \quad \text{if $w^{\mathrm{v}}_{l,m,n}(t+1)\ge0$}\\
S^{\mathrm{v}}_{l,m,n} = 0 & \quad \text{if $w^{\mathrm{v}}_{l,m,n}(t+1)<0$}\\
c^{\mathrm{p}}_{i_{\alpha},i_{\beta}}S^{\mathrm{p}}_{i_{\alpha},i_{\beta}}(t) \ge f^{\mathrm{pc}}_{i_{\alpha}, i_{\beta}}+\epsilon_2 & \quad \text{if $w^{\mathrm{p}}_{i_{\alpha},i_{\beta}}(t+1)\ge0$}\\
S^{\mathrm{p}}_{i_{\alpha},i_{\beta}}(t) = 0 & \quad \text{if $w^{\mathrm{p}}_{i_{\alpha},i_{\beta}}(t+1)<0$}
\end{align}
\end{subequations}

We assume that the distribution of the adjusted saturation flow, $\tilde{C}^{\mathrm{v}}_{l,m,n}(t)$, is fixed, so the corresponding mean value $\tilde{c}^{\mathrm{v}}_{l,m,n}(t)$ is also fixed. Substituting $\mathbf{S}$ into Equation \eqref{eq:alphas} obtains

\begin{subequations}
    \begin{align}
        &\omega^{\mathrm{v}}_1+\lambda\omega^{\mathrm{pc4}}\\
        \le&\sum_{l,m,n}\left[f^{\mathrm{v}}_{l,m,n}-\tilde{c}^{\mathrm{v}}_{l,m,n}(X^{\mathrm{v}}(t), \mathbf{S}(t))S^{\mathrm{v}}_{l,m,n}(t)\right]w^{\mathrm{v}}_{l,m,n}(t+1)+\sum_{i,\alpha,\beta}\left[f^{\mathrm{pc}}_{i_{\alpha}, i_{\beta}}-c^{\mathrm{p}}_{i_{\alpha},i_{\beta}}S^{\mathrm{p}}_{i_{\alpha},i_{\beta}}(t)\right]\lambda w^{\mathrm{p}}_{i_{\alpha}, i_{\beta}}(t+1)\\
        \le&\sum_{\{(l,m,n):w^{\mathrm{v}}_{l,m,n}(t+1)\ge0\}}-\epsilon_1w^{\mathrm{v}}_{l,m,n}(t+1)+\sum_{\{(l,m,n):w^{\mathrm{v}}_{l,m,n}(t+1)<0\}}f^{\mathrm{v}}_{l,m,n}w^{\mathrm{v}}_{l,m,n}(t+1)\\
        &+\sum_{\{(i,\alpha,\beta):w^{\mathrm{p}}_{i_{\alpha},i_{\beta}}(t+1)\ge0\}}-\lambda\epsilon_2w^{\mathrm{p}}_{i_{\alpha},i_{\beta}}(t+1)+\sum_{\{(i,\alpha,\beta):w^{\mathrm{p}}_{i_{\alpha},i_{\beta}}(t+1)<0\}}f^{\mathrm{pc}}_{i_{\alpha}, i_{\beta}}\lambda w^{\mathrm{p}}_{i_{\alpha},i_{\beta}}(t+1)\\
        \le& -\epsilon_1\sum_{l,m,n}|\mathbf{w^{\mathrm{v}}(t)}|-\lambda\epsilon_2 \sum_{i,\alpha,\beta}|\mathbf{w^{\mathrm{p}}(t)}|
    \end{align}
\end{subequations}

Since $\mathbf{w}^{\mathrm{v}}(t)$ and $\mathbf{w}^{\mathrm{p}}(t)$ are linear functions of $\mathbf{X}^{\mathrm{v}}(t)$ and $\mathbf{X}^{\mathrm{pc}}(t)$, respectively, it is easy to prove that there exist $\eta_1>0$ and $\eta_2>0$ such that

\begin{equation}
    |\mathbf{w}^{\mathrm{v}}(t)|\ge \eta_1 |\mathbf{X}^{\mathrm{v}}(t)|\quad \text{and}\quad
    |\mathbf{w}^{\mathrm{p}}(t)|\ge \eta_2 |\mathbf{x}^{\mathrm{pc}}(t)|
\end{equation}

Therefore, there exist $\xi_1>0$ and $\xi_2>0$ such that

\begin{equation}
    \alpha^{\mathrm{v}}_1+\alpha^{\mathrm{pc4}}\le -\xi_1|\mathbf{X^{\mathrm{v}}(t)}|-\xi_2|\mathbf{X^{\mathrm{pc}}(t)}|
\end{equation}

Thus, Equation \eqref{eq:key} holds under the control of PQ-MP.

\begin{lemma}\label{lemma:ms}
    If there exist $\xi_1>0$, $\xi_2>0$ such that the Lyapunov functions $|\mathbf{X}^{\mathrm{v}}(t)|^2$ and $|\mathbf{X}^{\mathrm{pc}}(t)|^2$ satisfy Equation \eqref{eq:key} under a control sequence, then the control sequence stabilizes the queue process.
\end{lemma}
\begin{proof}
    Taking expectation for both sides of Eq. \eqref{eq:key} obtains
    \begin{equation}\label{eq:lemma1}
        |\mathbf{X}^{\mathrm{v}}(t+1)|^2-|\mathbf{X}^{\mathrm{v}}(t)|^2+\lambda|\mathbf{X}^{\mathrm{pc}}(t+1)|^2-\lambda |\mathbf{X}^{\mathrm{pc}}(t)|^2\le k-\xi_1E[|\mathbf{X}^{\mathrm{v}}(t)|]-\xi_2E[|\mathbf{X}^{\mathrm{pc}}(t)|]
    \end{equation}
    Summing over $t=1, 2, 3, ..., T$ for both sides of Eq. \eqref{eq:lemma1} gives
    \begin{equation}\label{eq:lemma1}
        |\mathbf{X}^{\mathrm{v}}(T+1)|^2-|\mathbf{X}^{\mathrm{v}}(1)|^2+\lambda|\mathbf{X}^{\mathrm{pc}}(T+1)|^2-\lambda|\mathbf{X}^{\mathrm{pc}}(1)|^2\le k-\min(\xi_1, \xi_2)\sum_{t=1}^T\left(E[|\mathbf{X}^{\mathrm{v}}(t)|]+E[|\mathbf{X}^{\mathrm{pc}}(t)|]\right)
    \end{equation}
    Therefore,
    \begin{equation}\label{eq:lemma1}
        \frac{1}{T}\sum_{t=1}^T\left(E[|\mathbf{X}^{\mathrm{v}}(t)|]+E[|\mathbf{X}^{\mathrm{pc}}(t)|]\right)\le \frac{1}{\min(\xi_1, \xi_2)T}\left(k+|\mathbf{X}^{\mathrm{v}}(1)|^2+\lambda|\mathbf{X}^{\mathrm{pc}}(1)|^2\right)=M, \; \forall T
    \end{equation}
\end{proof}

Theorem \ref{theorem:ms} is proved.
\end{proof}
The maximum stability property indicates that PQ-MP can accommodate the combined vehicular and pedestrian demand if it can be served by any signal control strategy. Thus, the proposed PQ-MP maximizes network throughput.

\section{Numerical simulation}\label{sec:simulation}
This section demonstrates the control performance of PQ-MP under various conditions in microscopic traffic simulations using the SUMO simulation software \citep{SUMO2018}. The original queue-based MP that does not consider pedestrians, referred to as Q-MP in the following,  and the rule-based MP proposed in \citep{xu2023ped} are used as baseline models.

\subsection{Network setup}
A 5$\times$5 uniform grid network of two-way streets, as shown in Figure \ref{fig:net}, is used in the simulation. All links have three dedicated lanes, one for each of left-turn, through and right-turn vehicle movement. Each link has a sidewalk on the right side, and each intersection has a crossing on each of the four sides. The length for all links is 300 m. The speed limit for vehicles on all links is 15 m/s, and the average pedestrian travel speed is 1.3 m/s. 

\begin{figure}[htb!]
	\includegraphics[width=4.5in]{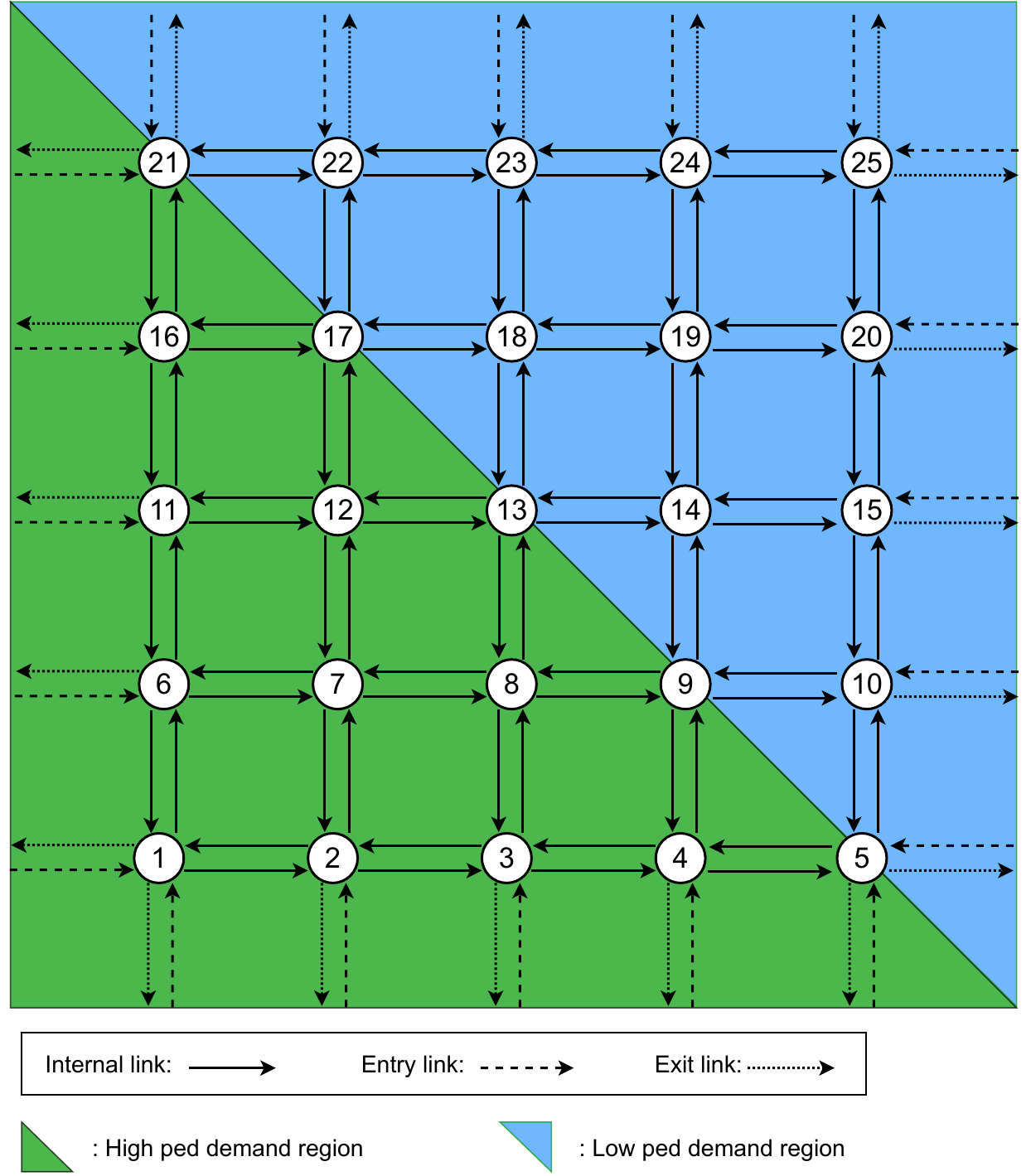}
	\centering
	\caption{Network setup.}
	\label{fig:net}
\end{figure}

All intersections have an identical set of admissible signal phases,  shown in Figure \ref{fig:signal}. There are 11 admissible phases in total. The first four rows in Figure \ref{fig:signal} correspond to the four phases for vehicle movements: northbound-southbound (NS) through (T) and right-turn (R) movement, NS left (L) movement, eastbound-westbound (EW) T and R movement and EW L movement. Pedestrians can only be served during the phases for T-R vehicular movements or during an exclusive phase for pedestrian-only movement, shown as the bottom phase in Figure \ref{fig:signal}. In the former situation, pedestrians from either one side or both sides can be served, and the right-turn vehicles need to yield to the conflicting pedestrian crosswalk movements. The time step for signal timing update is equal to the free-flow travel time of vehicles traversing a link, which is 20 s. This value is chosen due to two reasons: 1. it matches with the step size defined in the store-and-forward vehicle traffic dynamics, shown by Equation \eqref{eq:veh_transition}; 2. according to the Traffic Signal Timing Manual \citep{koonce2008traffic}, the minimum pedestrian green time should be equal to the time required by a pedestrian to cross the intersection plus 4 seconds. Under the simulation settings, the minimum pedestrian green time is equal to 18.76 s, which is close to the time step. The yellow time and all-red time are set to 3 s and 1 s, respectively.

\begin{figure}[htb!]
	\includegraphics[width=4.5in]{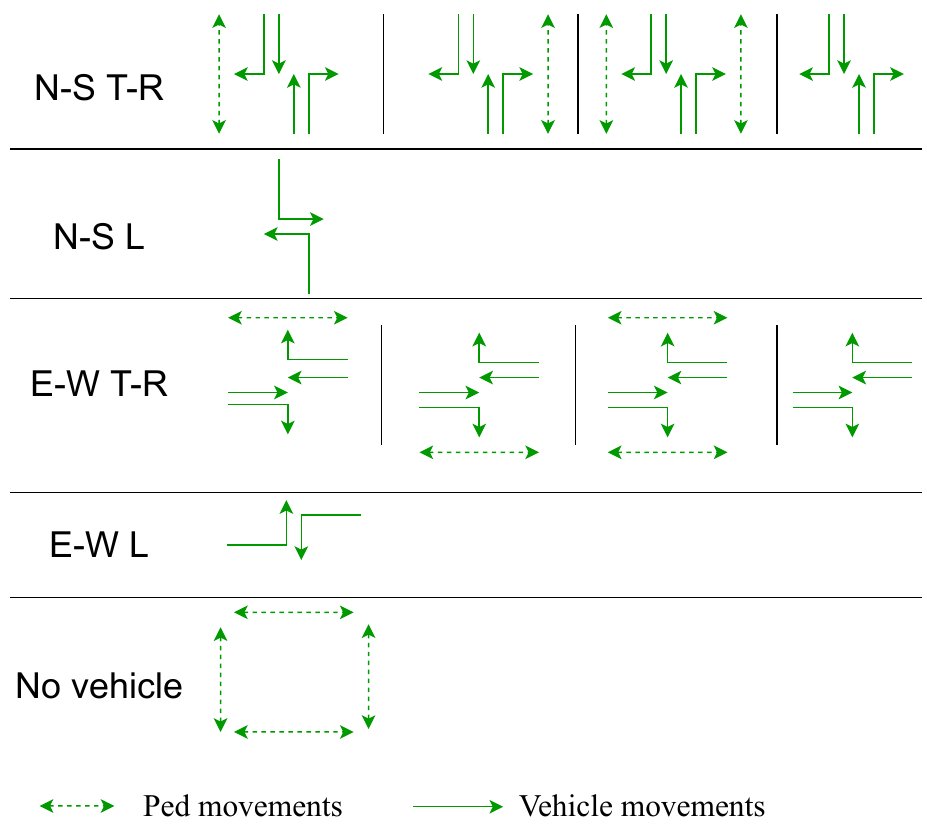}
	\centering
	\caption{Signal configuration.}
	\label{fig:signal}
\end{figure}

Vehicles enter and exit the network through the entry links and exit links, respectively, as shown in Figure \ref{fig:net}. For simplicity, identical vehicle demand value is used for all entry links, and four vehicle demand levels, i.e., \{400, 500, 600, 700\} veh/h/entry link, are tested. Each sidewalk serves as a pedestrian centroid to generate or attract trips. Every pedestrian (origin) centroid has a certain probability to generate a pedestrian trip to each of the rest (destination) centroid during the simulation. Different pedestrian demand levels are modeled. In Figure \ref{fig:net}, the pedestrian generation probability is 0.6 if both the origin and destination centroids are in the green region; otherwise, the probability between an OD pair is 0.3. 

The simulation time is 2 hours, and the second hour serves as a cool down period with zero demand. To consider the randomness, each simulation was run for 10 different starting seeds.

\subsection{Implementation of signal control algorithms}
We assume that the number of pedestrians waiting to cross the intersection is available for PQ-MP. The following values for $\lambda$: \{0.0005, 0.001, 0.002, 0.004, 0.006, 0.008, 0.01, 0.05, 0.1\}, in Equation \eqref{eq:pressure} are tested for the proposed PQ-MP.

For the rule-based MP, it is assumed that the waiting times of pedestrians in the pedestrian nodes are always available. A threshold value for the waiting time of the first arrival of crosswalk movements, denoted by $\tau$, is applied. At the instant of each signal update (every 20 s), if the waiting times from all crosswalk movements are lower than $\tau$, one of the four phases that include only vehicle movements shown in Figure \ref{fig:signal} is activated based on the Q-MP. Only vehicles are considered in this pressure calculation, and no pedestrian crosswalk movements are allowed. On the contrary, if the waiting time of certain crosswalk movements exceeds $\tau$, the phase serving the corresponding movements will be activated for the next phase. For example, if only the waiting time on the eastern side crosswalk movement exceeds the threshold, the second phase in the first row in Figure \ref{fig:signal} will be activated for the next step. For another example, if at least one crosswalk movement from each of the N-S and E-W directions incurs a waiting time longer than $\tau$, the no-vehicle phase, which is the bottom phase in Figure \ref{fig:signal}, will be activated. It is easy to imagine that a larger $\tau$ leads to a higher operational efficiency for vehicles, but generates higher delays for pedestrians.

For the Q-MP, the signal phase is determined by the pressure only taking vehicles into consideration. However, the associated pedestrian movements will always be activated as well, and the right-turn vehicles need to yield to pedestrians. For example, if the N-S T-R phase has the highest vehicle pressure, the third phase in the first row in Figure \ref{fig:signal} will be activated.

\subsection{Results}
\subsubsection{Stable regions}
As mentioned before, maximum stability is a desirable property of MP algorithms. Therefore, we first examined the stable region of each algorithm by monitoring the number of vehicles and pedestrians over time during the simulation. If a demand is stabilized under the control, the number of vehicles/pedestrians is expected to be stable after the network reaches equilibrium. Otherwise, the number of vehicles/pedestrians is expected to be continuously increasing until the demand starts decreasing. Figures \ref{fig:stable_region_qmp}--\ref{fig:stable_region_proposed} show the results under the control of the Q-MP, the rule-based MP and the PQ-MP, respectively. 

\begin{figure}[!htb]
	\includegraphics[width=4.5in]{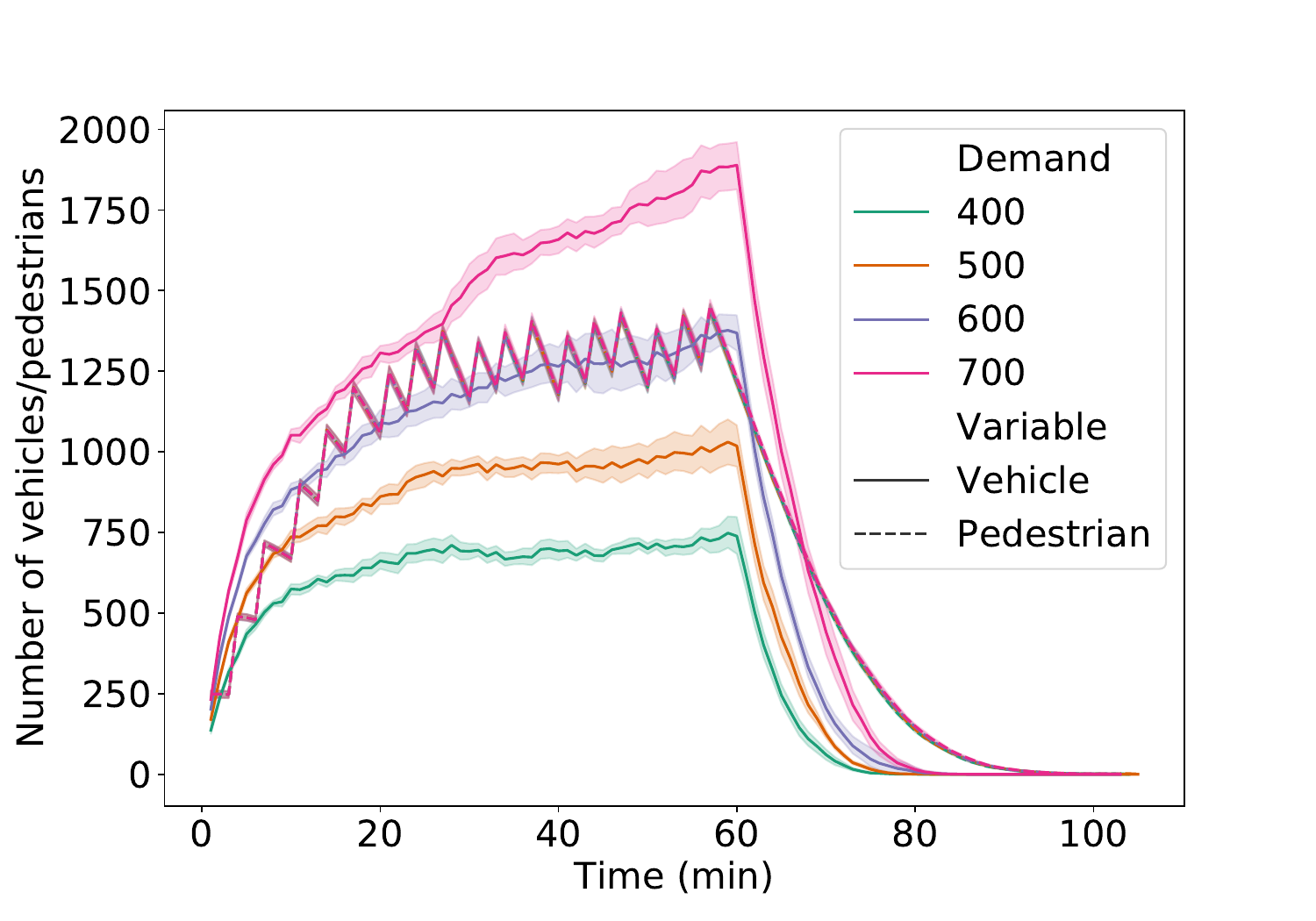}
	\centering
	\caption{Evolution of number of vehicles and pedestrians under Q-MP.}
	\label{fig:stable_region_qmp}
\end{figure} 

\begin{figure}[!ht]
    \centering
    \begin{subfigure}{0.48\textwidth}
        \centering
        \includegraphics[width=\textwidth]{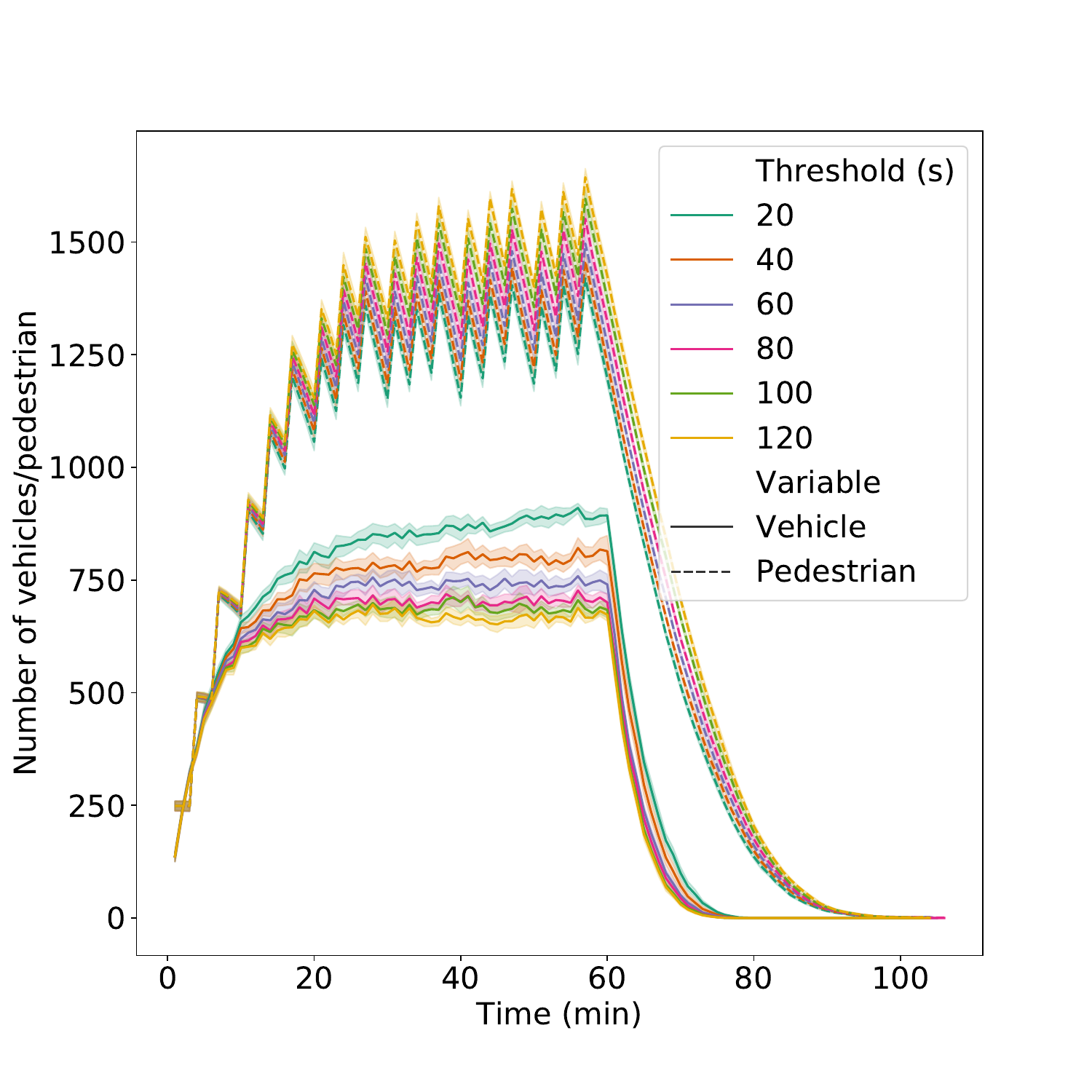}
        \caption{Vehicle demand: 400 veh/h/entry link}
    \end{subfigure}
    \begin{subfigure}{0.48\textwidth}
        \centering
        \includegraphics[width=\textwidth]{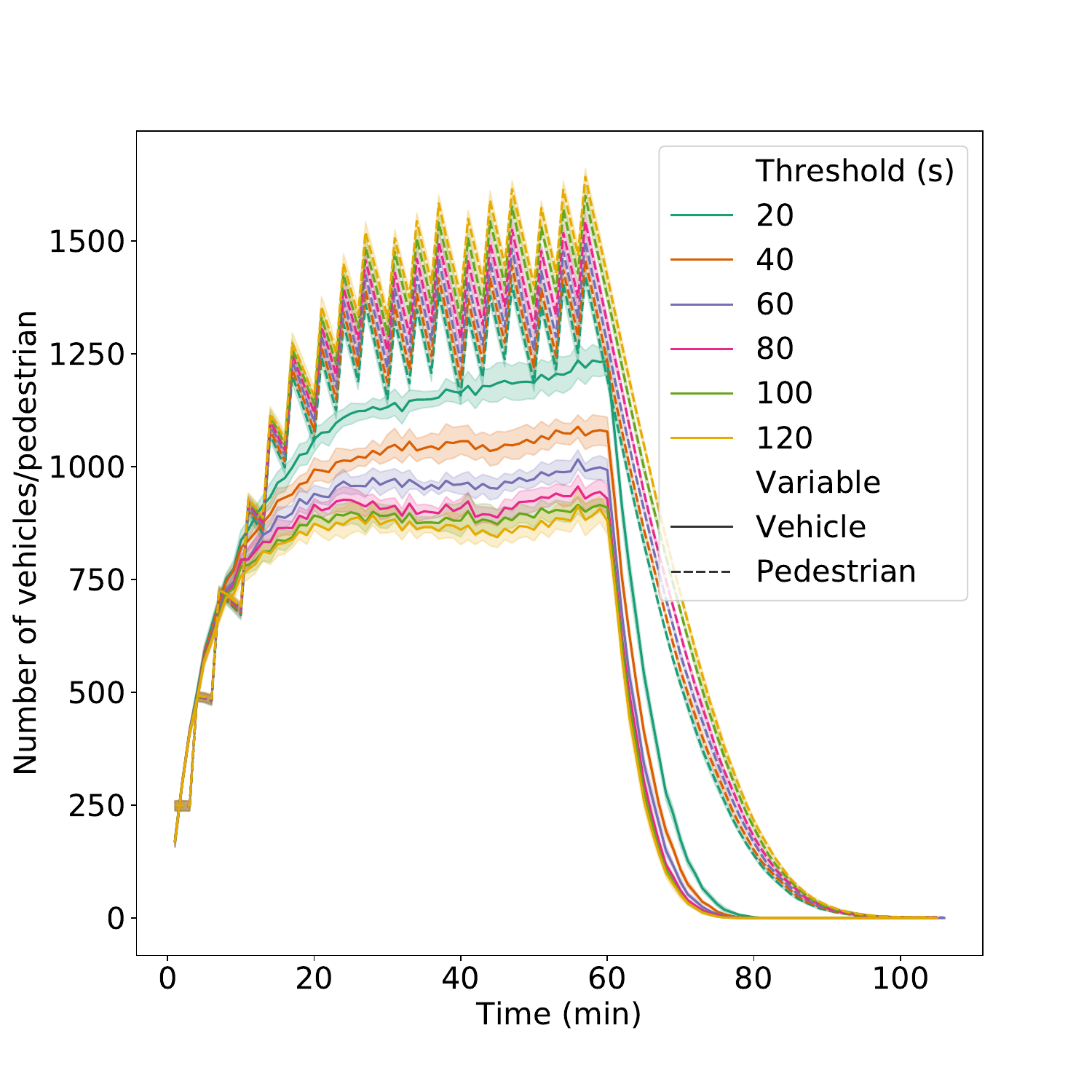}
        \caption{Vehicle demand: 500 veh/h/entry link}
    \end{subfigure}

    \begin{subfigure}{0.48\textwidth}
        \centering
        \includegraphics[width=\textwidth]{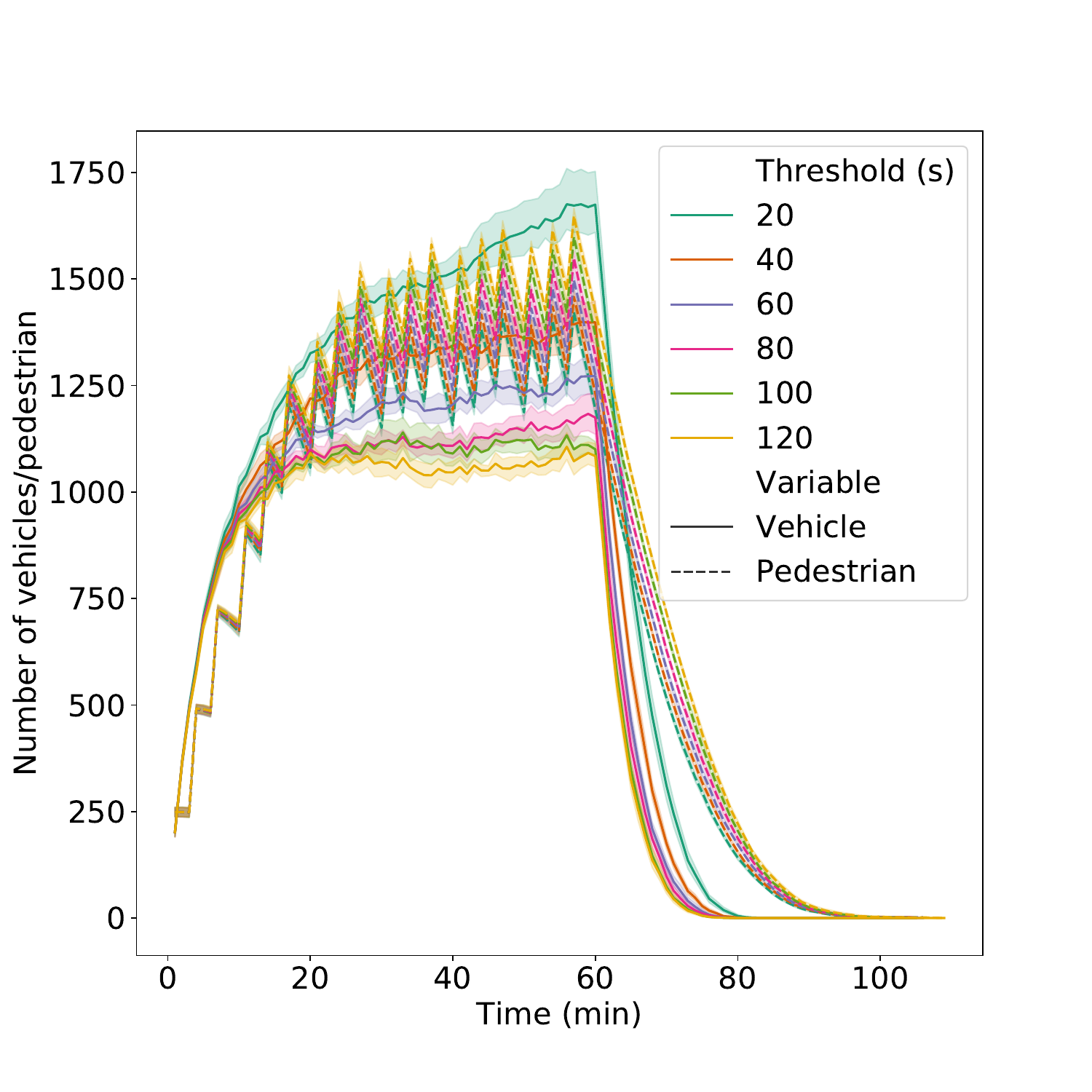}
        \caption{Vehicle demand: 600 veh/h/entry link}
    \end{subfigure}
    \begin{subfigure}{0.48\textwidth}
        \centering
        \includegraphics[width=\textwidth]{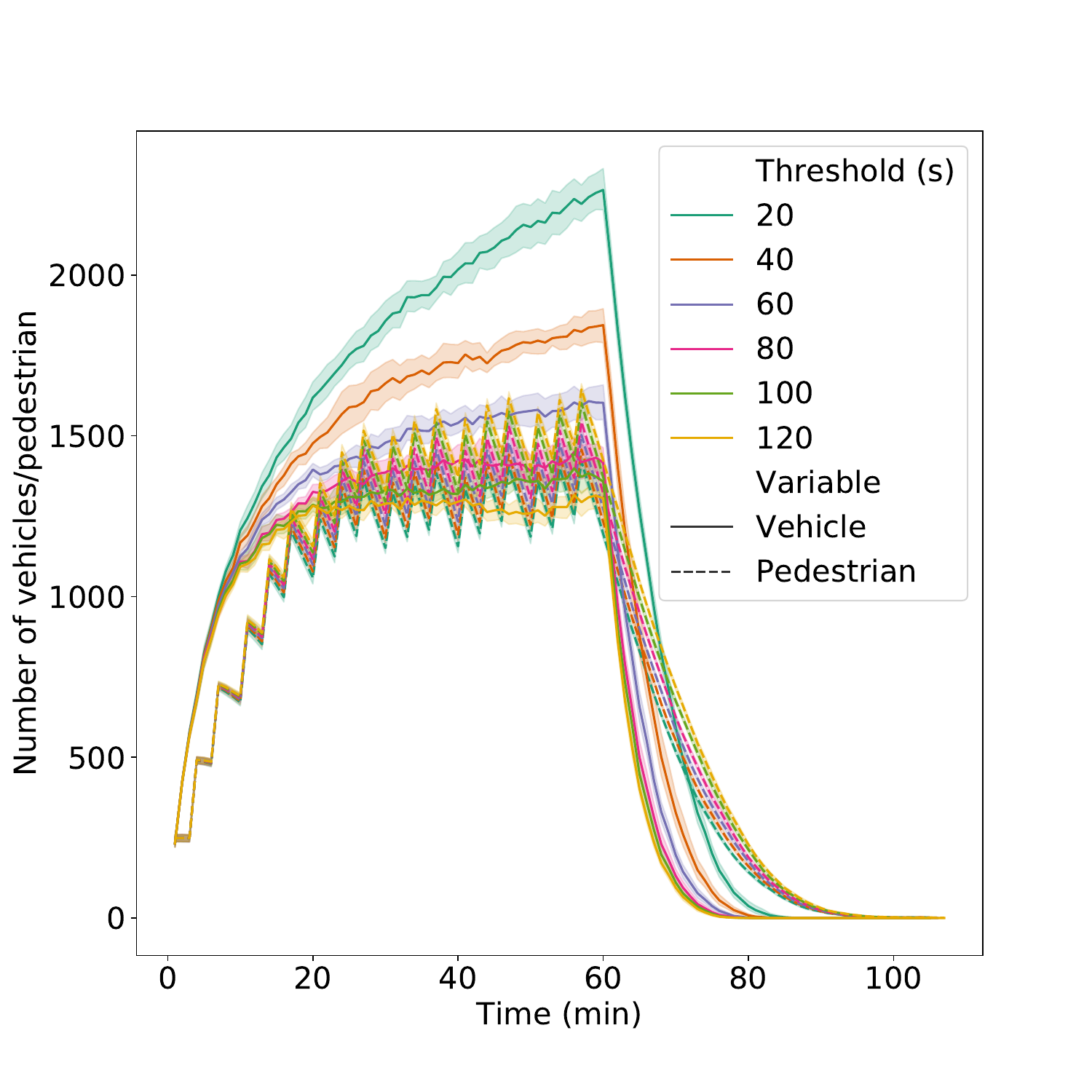}
        \caption{Vehicle demand: 700 veh/h/entry link}
    \end{subfigure}
    \caption{Evolution of number of vehicles and pedestrians under rule-based MP.}
    \label{fig:stable_region_rule}
\end{figure}

\begin{figure}[!ht]
    \centering
    \begin{subfigure}{0.48\textwidth}
        \centering
        \includegraphics[width=\textwidth]{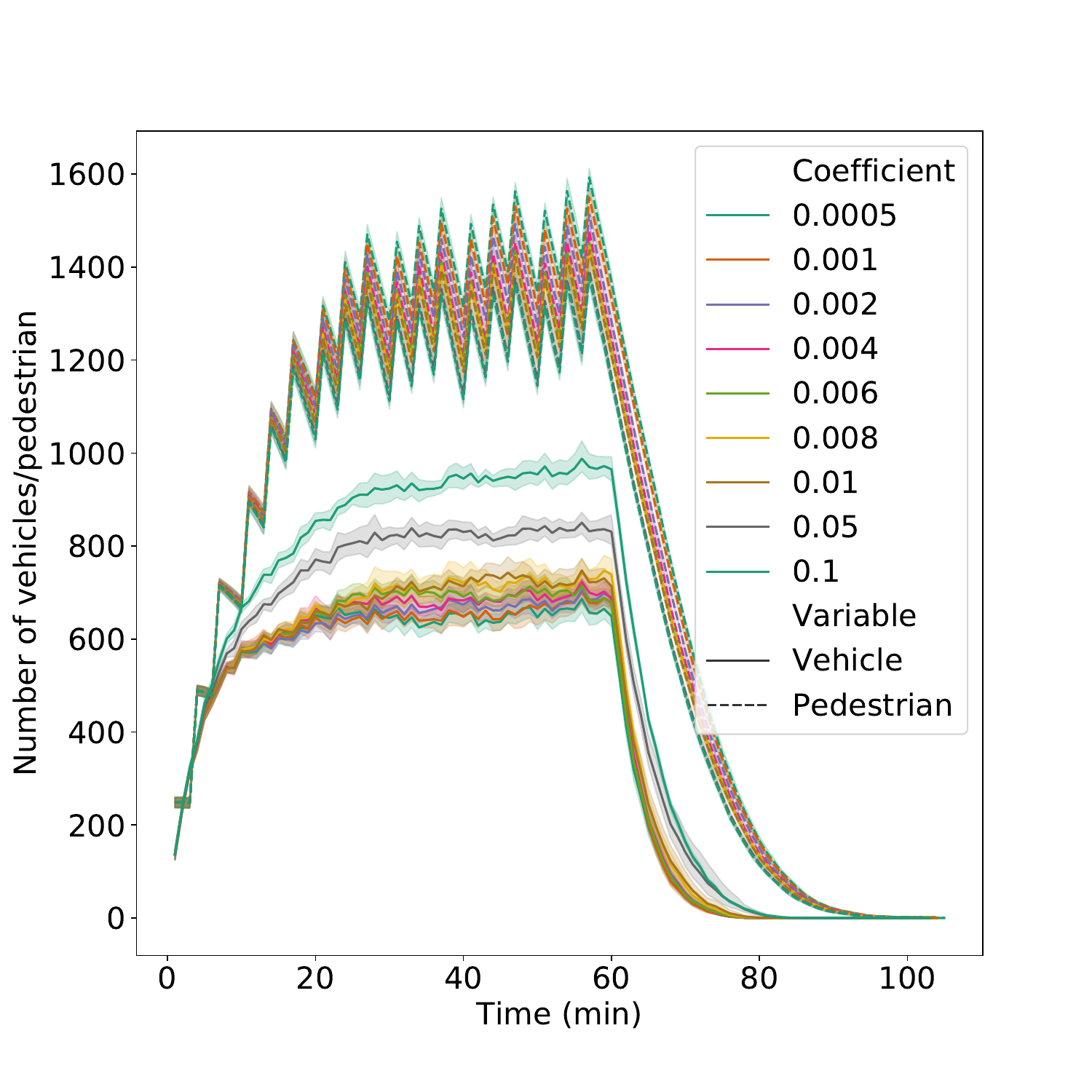}
        \caption{Vehicle demand: 400 veh/h/entry link}
    \end{subfigure}
    \begin{subfigure}{0.48\textwidth}
        \centering
        \includegraphics[width=\textwidth]{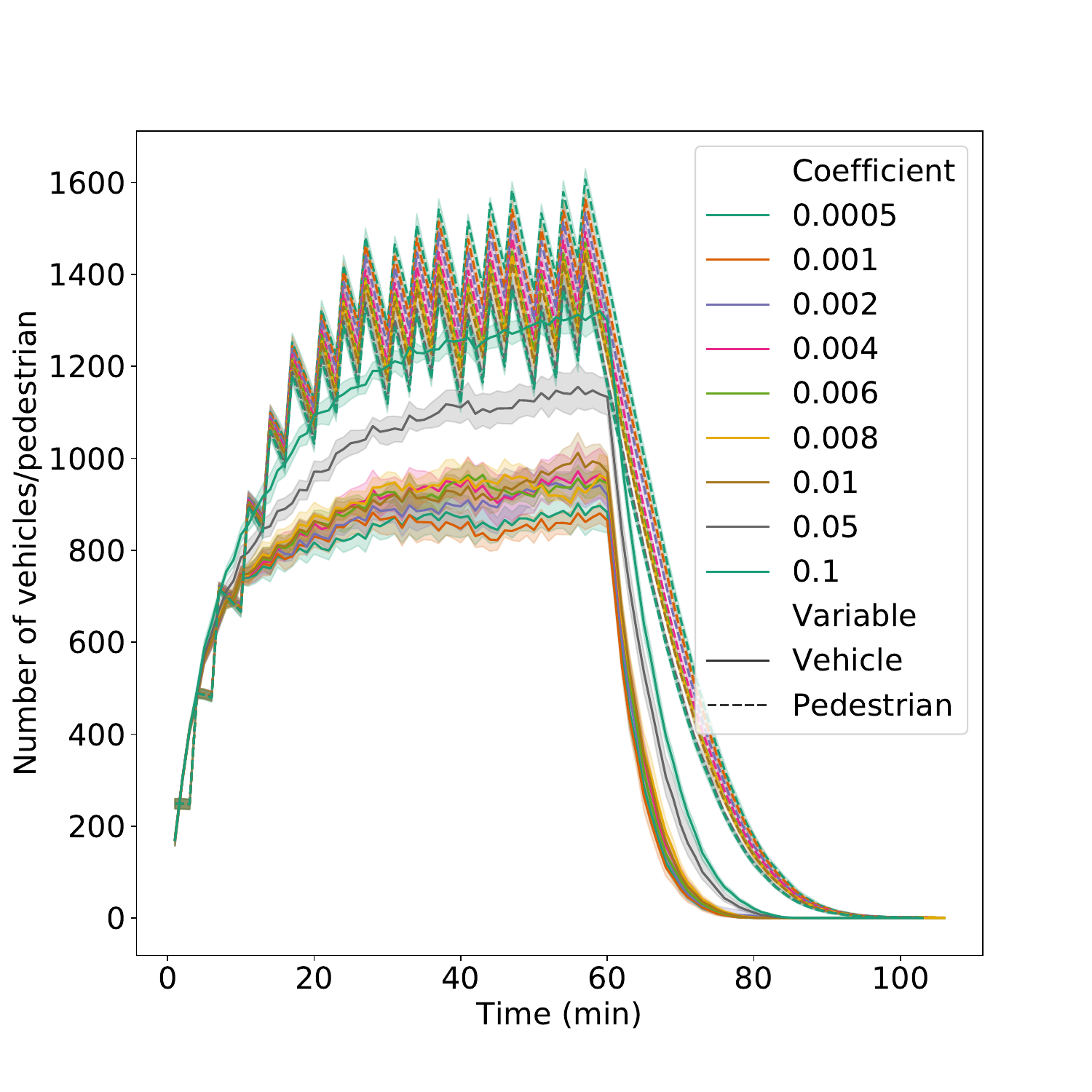}
        \caption{Vehicle demand: 500 veh/h/entry link}
    \end{subfigure}

    \begin{subfigure}{0.48\textwidth}
        \centering
        \includegraphics[width=\textwidth]{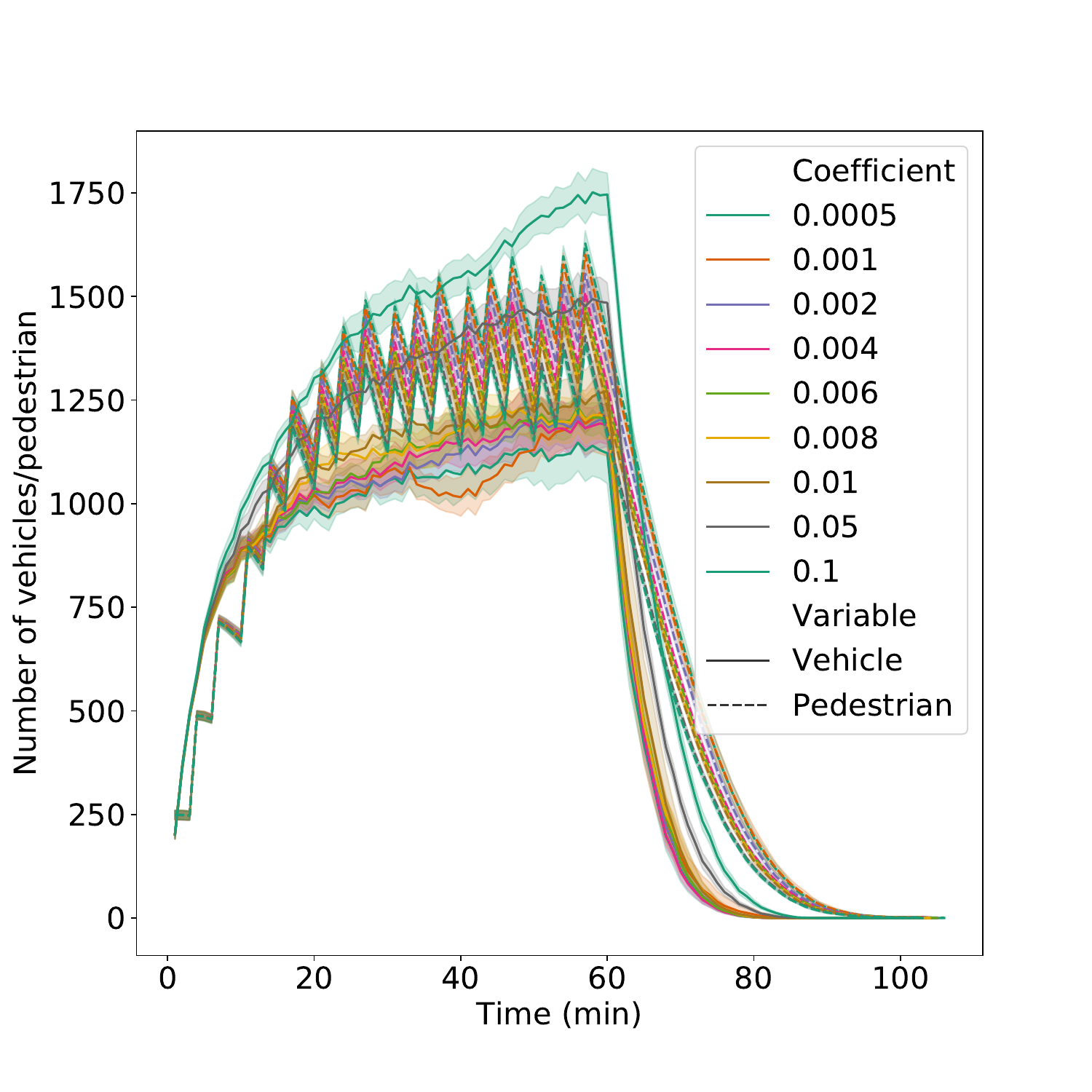}
        \caption{Vehicle demand: 600 veh/h/entry link}
    \end{subfigure}
    \begin{subfigure}{0.48\textwidth}
        \centering
        \includegraphics[width=\textwidth]{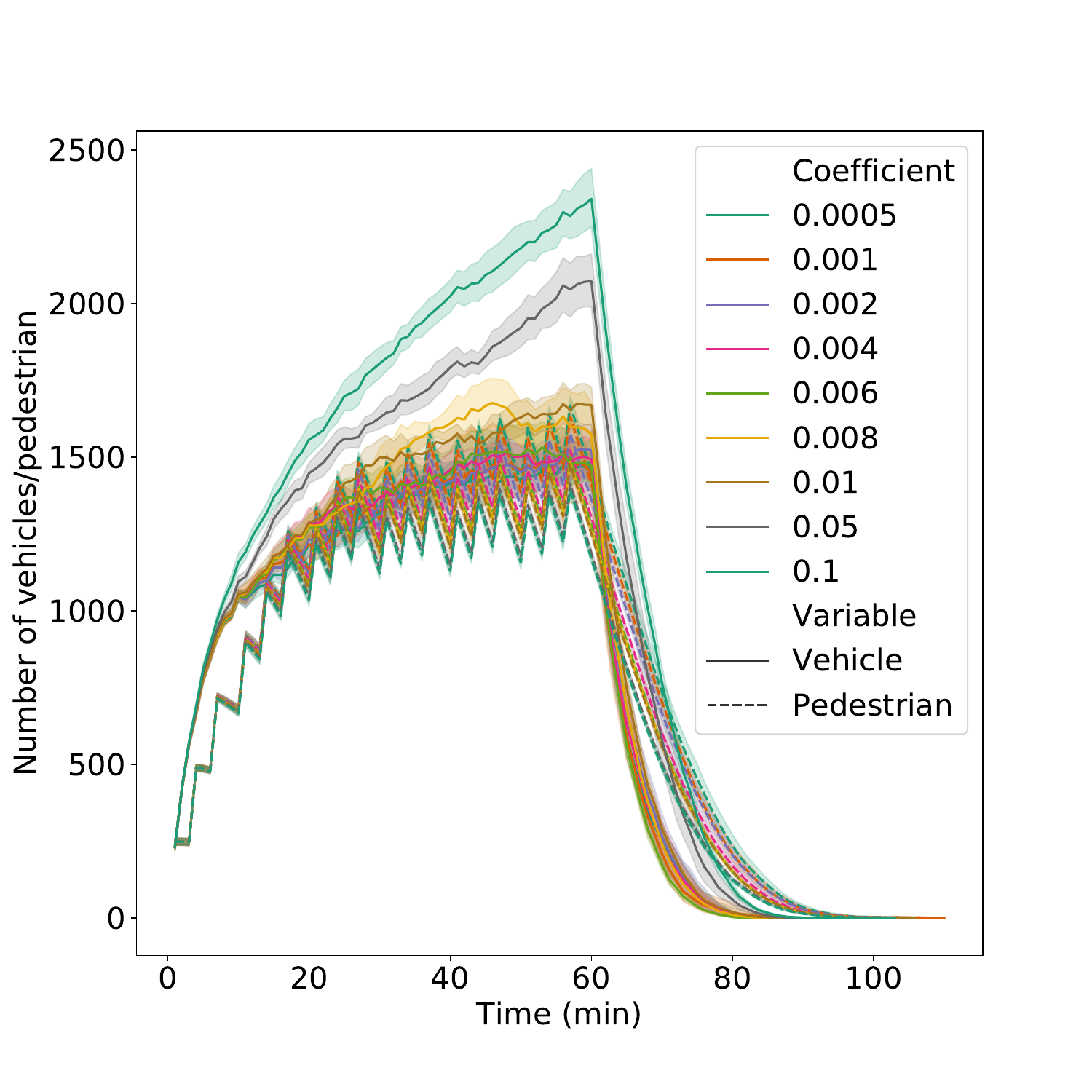}
        \caption{Vehicle demand: 700 veh/h/entry link}
    \end{subfigure}
    \caption{Evolution of number of vehicles and pedestrians under PQ-MP.}
    \label{fig:stable_region_proposed}
\end{figure}

Figure \ref{fig:stable_region_qmp} shows that  the variation of pedestrian evolution across different vehicle demand levels is insignificant compared to the variation in vehicle evolution. In addition to the fixed pedestrian demand, another reason is that under this control policy, pedestrians are always served during the T-R phases with higher priorities than the right-turn vehicles. Therefore, as long as the turning ratios do not change, the average service rate for pedestrians is expected to be relatively constant as well. On the other hand, since the right-turn vehicles are required to yield to pedestrians, under the fixed turning ratios, an increase in the vehicle demand level leads  increases in delays for the right-turn vehicles. The evolution of the number of vehicles indicates that Q-MP can stabilize vehicle demands under 500 veh/h/entry link. Once the vehicle demand exceeds 600 veh/h/entry link, the number of vehicles continues to increase until the vehicle demand drops to zero at $t=60$ min. Compared to Q-MP, both the proposed PQ-MP and the rule-based MP can still stabilize those vehicle demands with certain values for $\lambda$ and $\tau$, respectively, as shown in Figures \ref{fig:stable_region_rule}--\ref{fig:stable_region_proposed}. This finding unveils the necessity of the development of MP algorithms considering pedestrians for the network stability.

For the rule-based MP algorithm, the threshold for the waiting time determines the frequency that pedestrians receive green times. A lower threshold of pedestrian waiting time allows pedestrians to be served more frequently and thus, increases the number of vehicles in the network. The results in Figure \ref{fig:stable_region_rule} show that the rule-based MP algorithm can achieve a larger stable region for vehicles than the Q-MP if the threshold is large enough. For example, the rule-based model can even stabilize the vehicle demand of 700 veh/h/entry link if the threshold is larger than 60 s.

For the proposed PQ-MP, as $\lambda$ increases, the pedestrian-related term in Equation \eqref{eq:pressure} makes higher contribution to the pressure calculation and the influence of the pedestrian states on the control outputs becomes greater. As a result, the number of vehicles in the network increases as $\lambda$ rises. For the lowest vehicle demand scenario, all coefficients can stabilize the vehicle demand level, as shown by Figure \ref{fig:stable_region_proposed}. As the vehicle demand level increases, the number of coefficients that can stabilize the network decreases. For example, when the entry demand is 600 veh/h/entry link, only the coefficients less than 0.002 can stabilize the network. When the vehicle demand reaches 700 veh/h/entry link, the network cannot be stabilized by any of the examined coefficients. Note although the rule-based model can stabilize this vehicle demand when the threshold value for the pedestrian waiting time is larger than 60 s, this superiority over PQ-MP is not surprising. Theoretically, under the control of the rule-based MP, the stable region for vehicles can be maximized when the threshold is infinity. However, this advantage is achieved with the sacrifice in the operational efficiency of pedestrians. This conflicts with the goal of the proposed model, which is to ensure the overall operational efficiency for both vehicles and pedestrians.

\subsubsection{Travel delay}
In addition to the stable region, travel delay---the most commonly used metric to evaluate the control efficiency in traffic operations---obtained from all models are examined in this section. Figure \ref{fig:delay_rule} and Figure \ref{fig:delay_propose} show the average delay and total delay for both vehicles and pedestrians under the control of the rule-based model and the PQ-MP, respectively.

\begin{figure}[!ht]
    \centering
    \begin{subfigure}{0.8\textwidth}
        \centering
        \includegraphics[width=\textwidth]{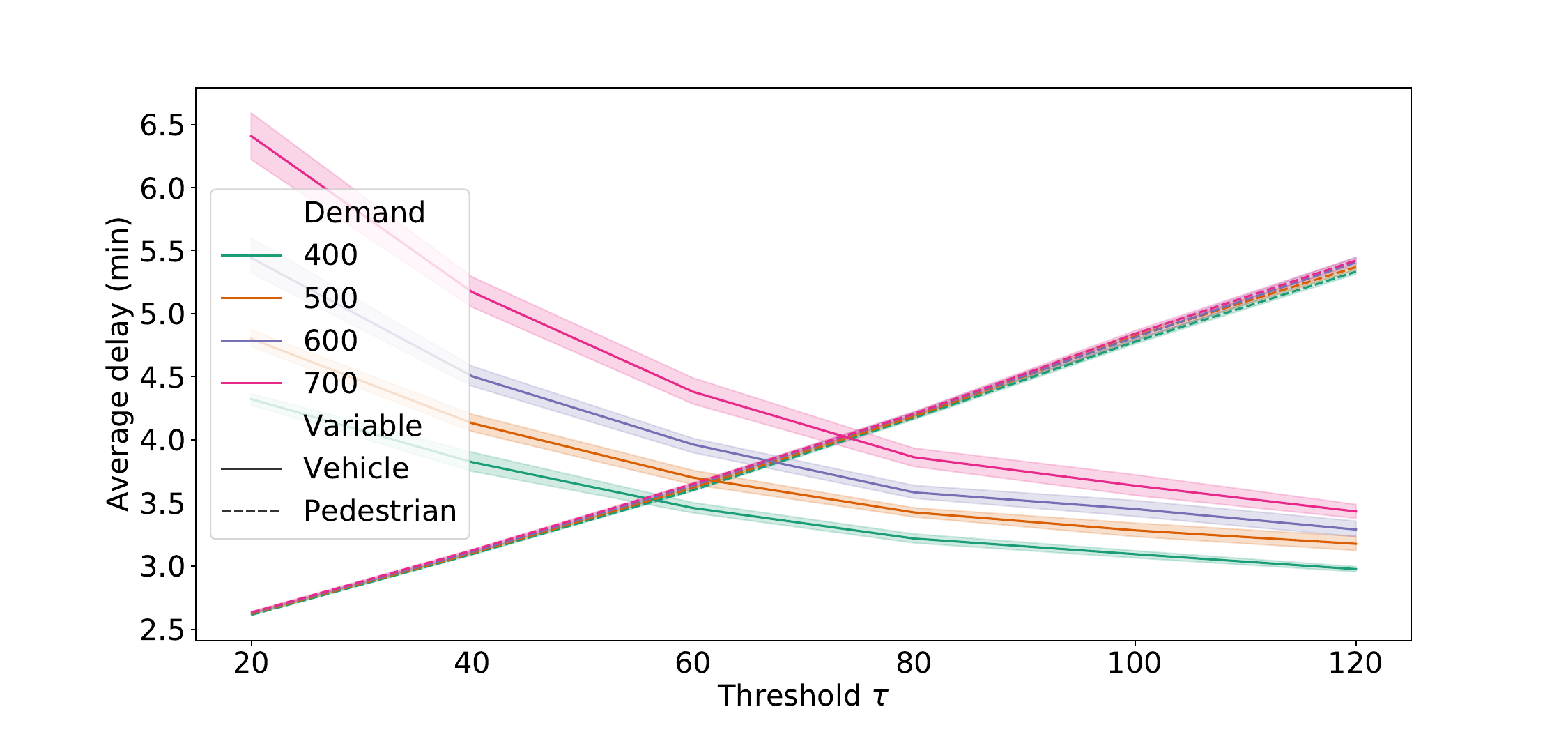}
        \caption{Average delay (min)}
    \end{subfigure}
    
    \begin{subfigure}{0.8\textwidth}
        \centering
        \includegraphics[width=\textwidth]{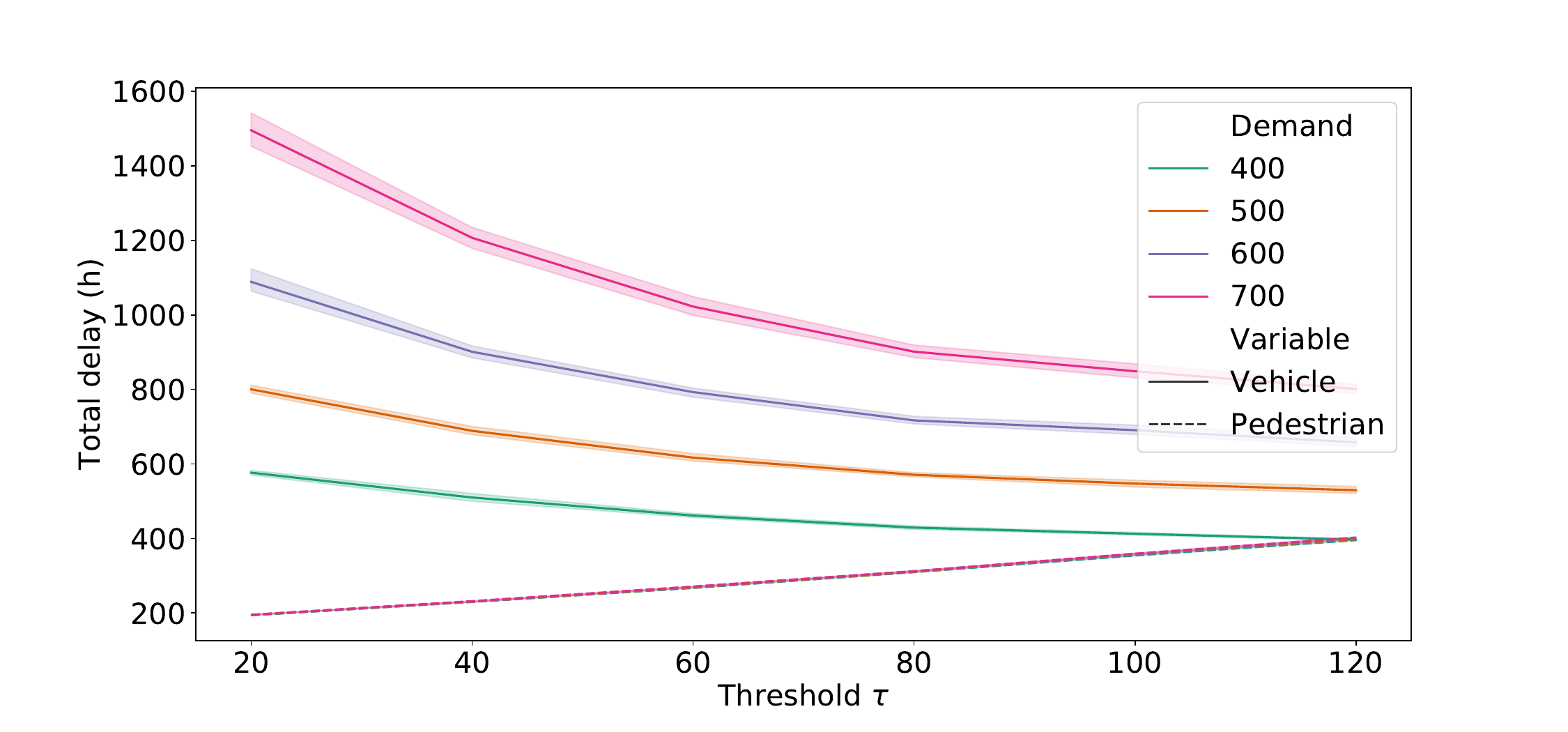}
        \caption{Total delay (h)}
    \end{subfigure}
    \caption{Delay from the rule-based MP algorithm.}
    \label{fig:delay_rule}
\end{figure}

\begin{figure}[!ht]
    \centering
    \begin{subfigure}{0.8\textwidth}
        \centering
        \includegraphics[width=\textwidth]{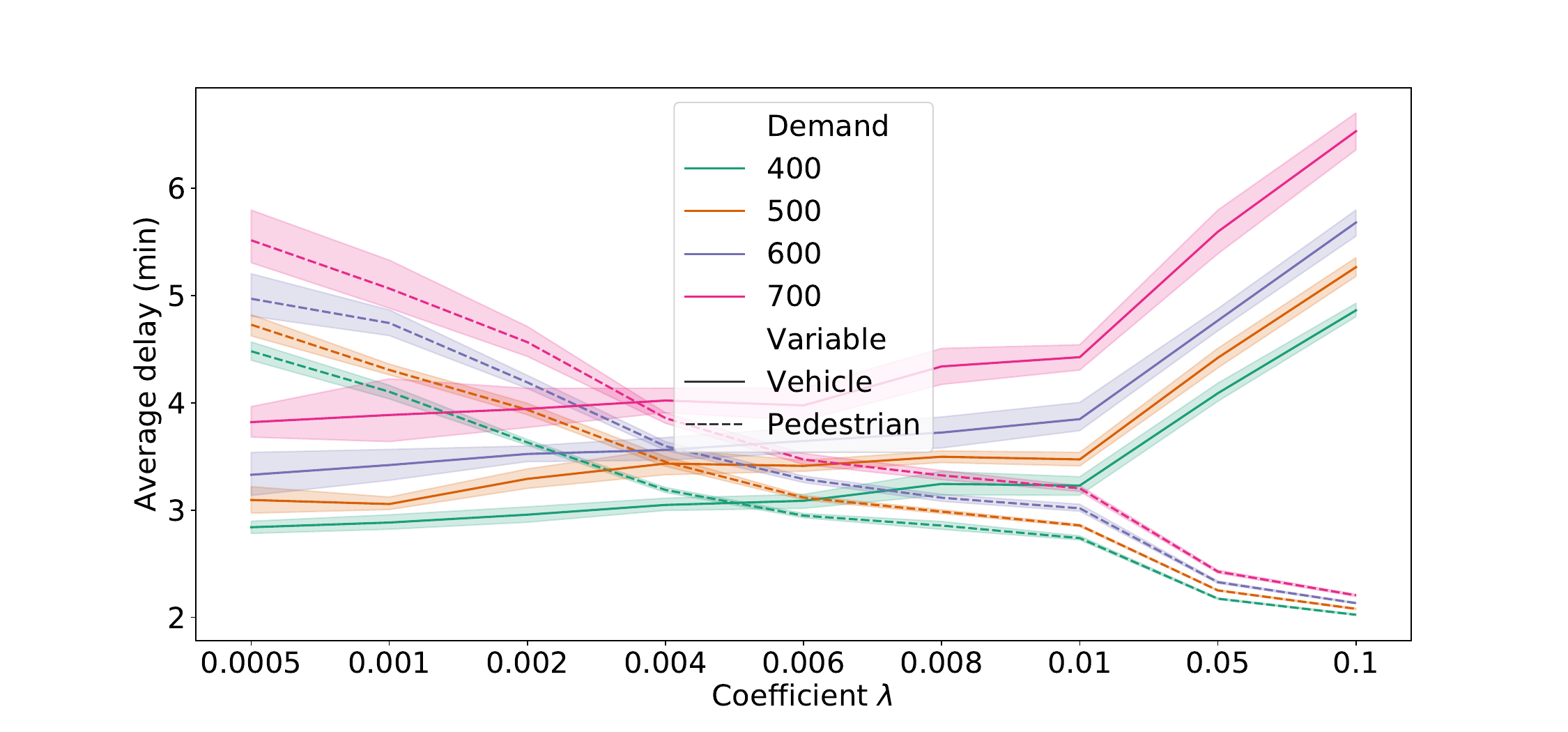}
        \caption{Average delay (min)}
    \end{subfigure}
    
    \begin{subfigure}{0.8\textwidth}
        \centering
        \includegraphics[width=\textwidth]{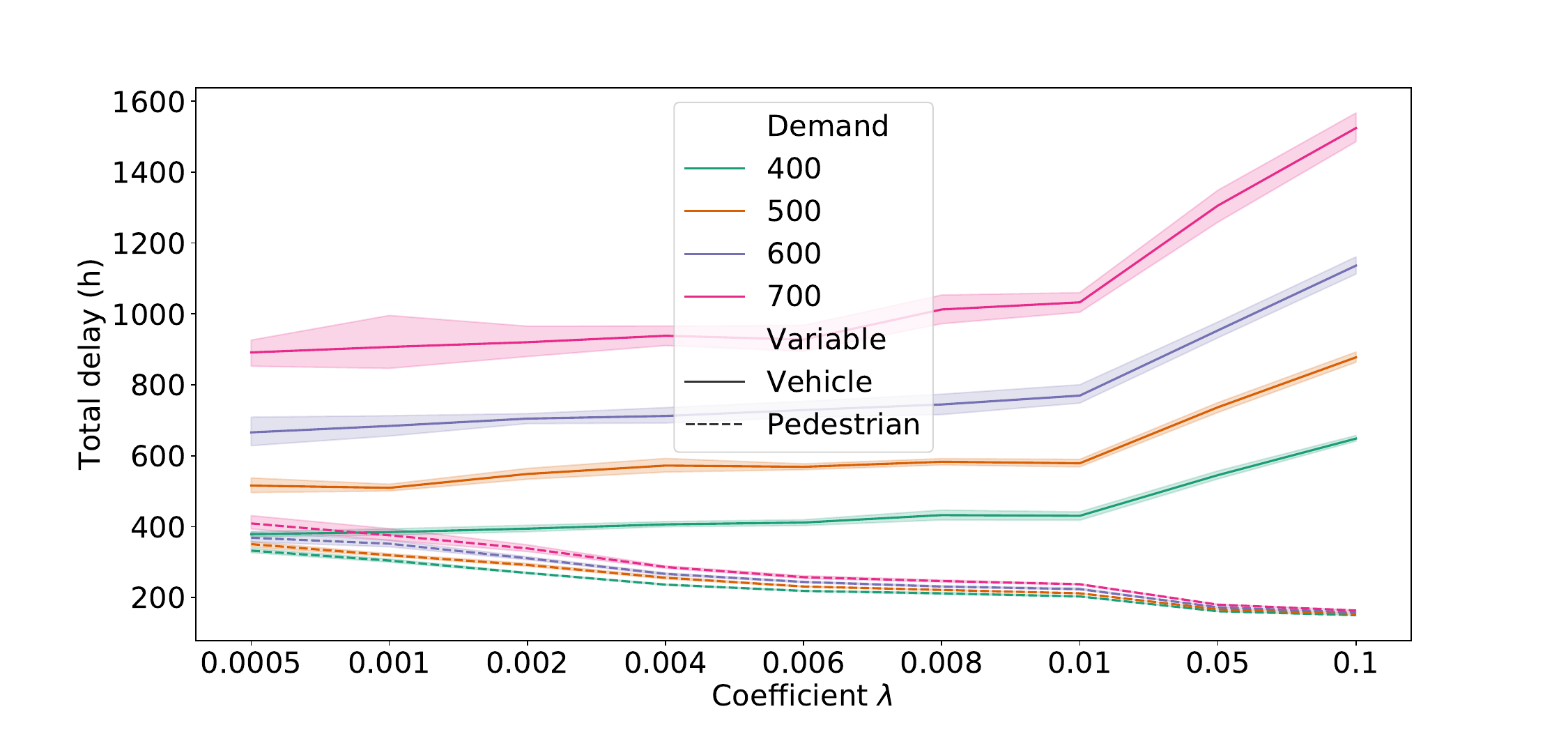}
        \caption{Total delay (h)}
    \end{subfigure}
    \caption{Delay from PQ-MP.}
    \label{fig:delay_propose}
\end{figure}

Figure \ref{fig:delay_rule} reveals that, under the rule-based MP algorithm, the delay of vehicles decreases as the threshold for the pedestrian waiting time increases, and the delay incurred by pedestrians is close to a linear function of the threshold. The reason is that a pedestrian phase is activated when the largest delay incurred by all pedestrians served by this phase is equal to the threshold. Therefore, the threshold can be regarded as an approximation of the largest delay that can be incurred by a pedestrian at an intersection. Assume that the pedestrian arrivals are random, then the average/total delay of pedestrians is linear to the largest delay, which is equivalent to the threshold value. 

Figure \ref{fig:delay_propose} shows that when $\lambda$ is lower than 0.01, the delay of vehicles resulting from the PQ-MP is relatively stable against the change in the coefficient. However, the delay of pedestrians can be reduced significantly by increasing the coefficient. When the coefficient exceeds 0.01, the delay of vehicles increases rapidly. 

The actual delay for the Q-MP is omitted since it has the worst overall performance. However, it is used as the baseline for the following comparison between PQ-MP and the rule-based MP.

In order to generate a fair comparison, we assume the average occupancy of vehicles is 1.3 pax/veh and calculate the total person delay to evaluate the performance of PQ-MP for all coefficients and the rule-based MP algorithm for all threshold values under all vehicle demand levels. We found that the coefficient of 0.0006 and threshold of 80 s generate the lowest person delay for both algorithms, respectively. Therefore, only the results from the corresponding models are compared. Note that these values do not necessarily generate the best performance for other scenarios. In practice, the values of both coefficients are determined by individual traffic engineers based on how they value the operational efficiencies between vehicles and pedestrians. For a better visualization, the delay reductions from both models compared to Q-MP are shown in Figure \ref{fig:delay_improve}. A positive reduction in these figures represents the algorithm can reduce travel delay compared to the Q-MP, and vice versa. These results manifest that although the rule-based model can achieve a desirable performance for vehicles, the overall efficiency cannot be ensured. By considering the traffic states of both vehicles and pedestrians in the pressure calculation, the PQ-MP manages to improve the overall efficiency. This improvement is more significant for relatively low vehicle demand scenarios. Note that if we select a lower threshold for the waiting time, e.g., 60 s, both the vehicle delay and pedestrian delay from the rule-based model are higher than those of the PQ-MP. 

\begin{figure}[!ht]
    \centering
    \begin{subfigure}{0.8\textwidth}
        \includegraphics[width=4.5in]{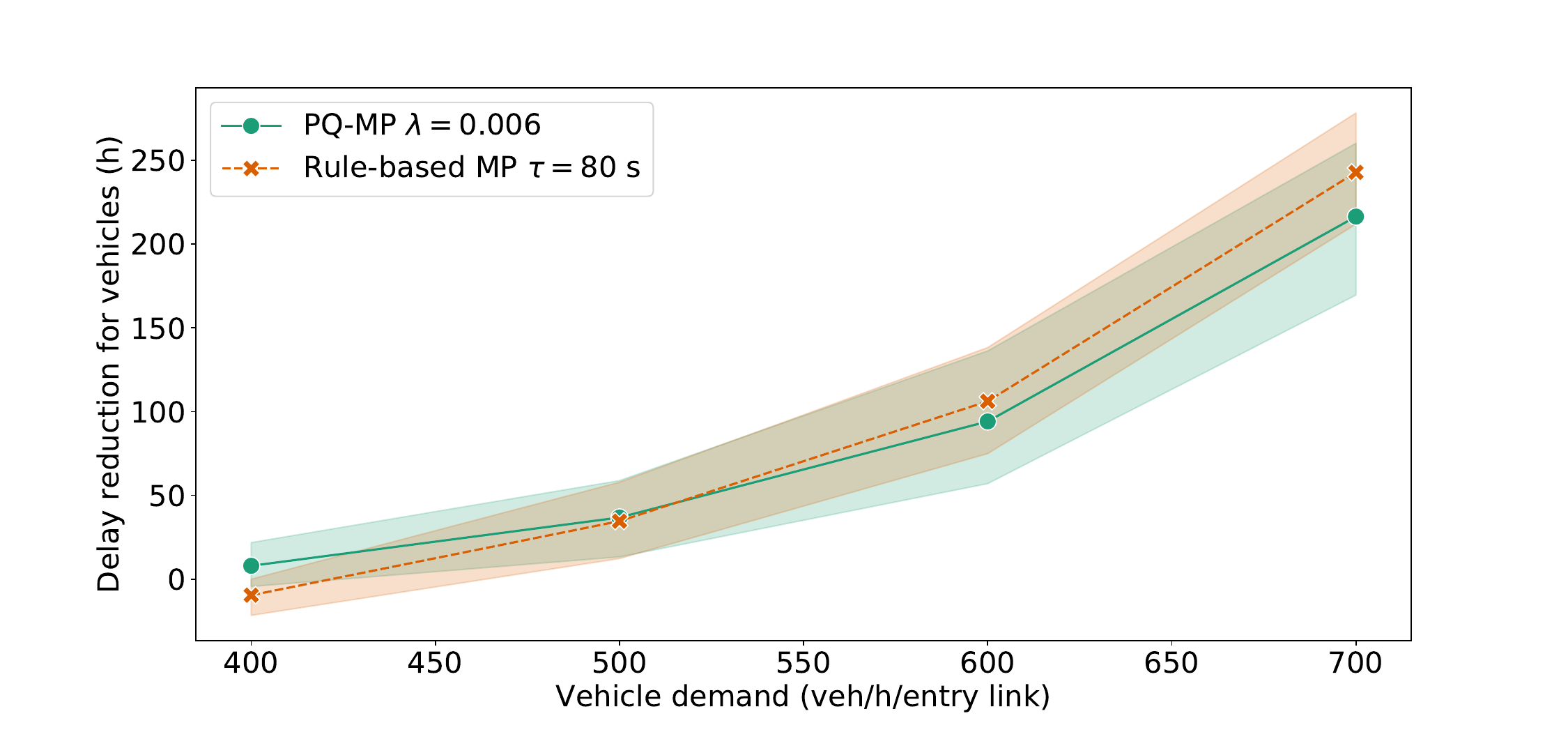}
    	\centering
    	\caption{Vehicle delay reduction.}
    	\label{fig:veh_improve}
    \end{subfigure}
    
    \begin{subfigure}{0.8\textwidth}
        \includegraphics[width=4.5in]{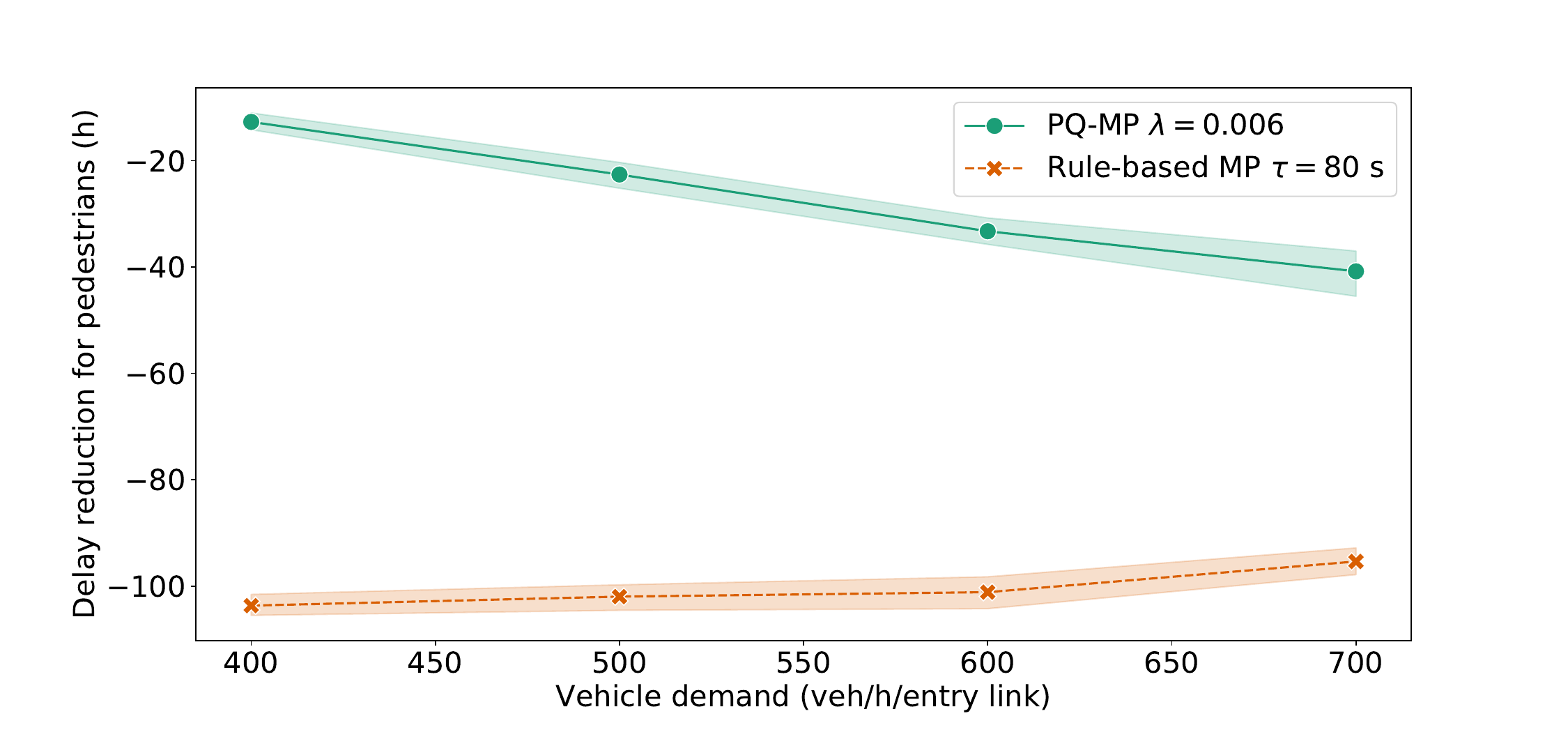}
    	\centering
    	\caption{Pedestrian delay reduction.}
    	\label{fig:ped_improve}
    \end{subfigure}

    \begin{subfigure}{0.8\textwidth}
        \includegraphics[width=4.5in]{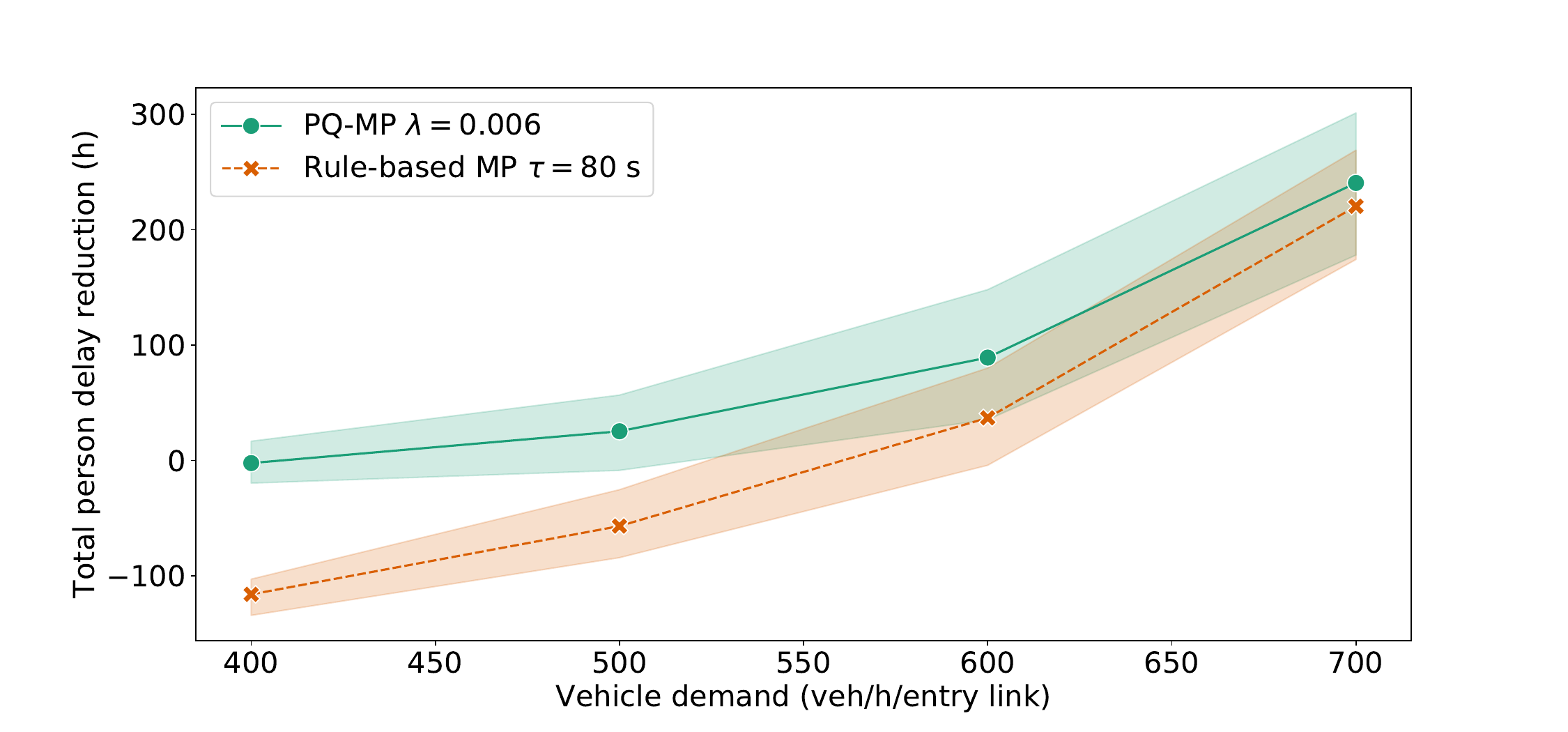}
    	\centering
    	\caption{Person delay reduction.}
    	\label{fig:total_improve}
    \end{subfigure}
    \caption{Delay reduction compared to Q-MP.}
    \label{fig:delay_improve}
\end{figure}

Figure \ref{fig:veh_improve} indicates that compared to the Q-MP, both models can reduce the vehicle delay significantly, especially for the large vehicle demand scenarios. Moreover, the rule-based model generates a lower vehicle delay than the proposed algorithm when vehicle demand exceeds 600 veh/h/entry link, which is consistent with the findings on the stable region. On the contrary, both models increase the pedestrian delay compared to Q-MP, as shown in Figure \ref{fig:ped_improve}. This is reasonable due to the same reason for the smaller vehicle stable region of Q-MP mentioned in the previous section. 

The proposed PQ-MP produces considerably lower pedestrian delay than the rule-based model, as shown in Figure \ref{fig:ped_improve}. Consequently, the person delay from PQ-MP is lower than the delay from the rule-based model under all tested vehicle demand levels, as shown in Figure \ref{fig:total_improve}. In addition, due to the large delay incurred by pedestrians, the person delay from the rule-based MP is significantly higher than the Q-MP when the vehicle demand is under 500 veh/h/entry link. The PQ-MP generates similar person delay to the Q-MP when the vehicle demand is 400 veh/h/entry link and significantly lower delay for all other vehicle demand levels. Compared to rule-based MP, the PQ-MP can reduce the total person delay by 114 hours when the vehicle demand is 400 veh/h/entry link. The reduction in person delay decreases with increasing vehicle demand. It reaches 20 hours when the vehicle demand is equal to 700 veh/h/entry link. 

\subsubsection{Influence of pedestrian demand on vehicle delays}

We further explore the vehicle delay incurred in the high (green) and low (blue) pedestrian demand regions shown in \ref{fig:net} to investigate the impact of pedestrian demand level on vehicle delays from both the rule-based MP and the PQ-MP algorithms. The delay reductions compared to Q-MP, which still serves as the baseline algorithm, are shown in Figure \ref{fig:veh_improve_pedregion}. It shows that, consistent with the pattern of the entire network shown in Figure \ref{fig:veh_improve}, the delay reductions from both models in both region increase with the vehicle demand, and the vehicle delay reduction from the PQ-MP, shown by the circle markers, is higher than that from the rule-based MP, shown by the cross markers, in both regions when the vehicle demand is less than 500 veh/h/entry link, and different pattern occurs when the vehicle demand exceeds 600 veh/h/entry link. In addition, two other interesting phenomena are observed.

\begin{figure}[!htb]
	\includegraphics[width=6in]{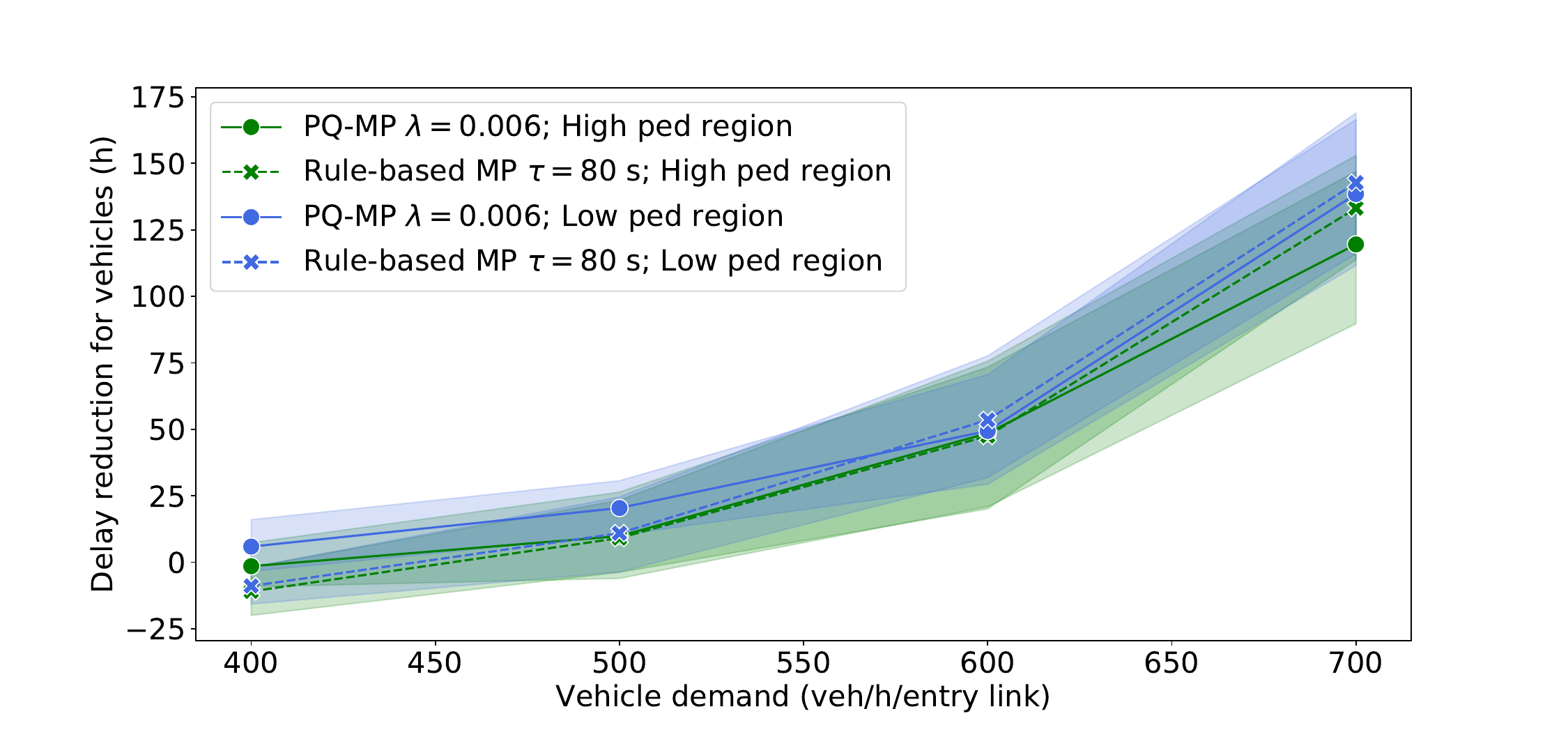}
	\centering
	\caption{Impact of pedestrian demand on vehicle delay.}
	\label{fig:veh_improve_pedregion}
\end{figure} 

First, when the vehicle demand is very high, i.e., 700 veh/h/entry link, the difference of delay reduction between both models in the high pedestrian demand region, shown by the green markers, is larger than that in the low pedestrian demand region, shown by the blue markers. The rule-based MP employs the same threshold value for both regions. However, in PQ-MP, the pedestrian term has a higher proportion in the phase pressure, shown by Equation \eqref{eq:pressure}, in the high pedestrian demand region than that in the low pedestrian demand region. Consequently, under the control of PQ-MP, the phases serving pedestrians are activated more frequently, and thus, the vehicle delay is higher in the high demand region, which results in a larger difference of delay reduction between both models.

Second, for both models, the delay reduction in the region with low pedestrian demand region, shown by the blue markers, is higher than that in the region with high pedestrian demand region, shown by the green markers. Note that since all entry links have the same external inflows, and the turning ratios at all intersections are identical, the patterns of average vehicle flows in both regions are similar. In addition, because the Q-MP only considers vehicle states, in the long run, the signal timing of one intersection, in terms of the average service rate for each vehicle movement at that intersection, is expected to be similar to that of the symmetric intersection with respect to the common boarder of both regions. Unlike Q-MP, both the rule-based MP and PQ-MP tend to allocate less green time to pedestrians in the low pedestrian demand region than the other region. For the rule-based model, the average pedestrian arrival rate at an intersection is lower in the low pedestrian demand region, so the intervals between pedestrian signal activation condition being satisfied are longer. Therefore, the pedestrian phases are activated less frequently in the low pedestrian demand region. Consequently, the vehicle delay improvement in that region from the rule-based MP is higher. Similarly, when the pedestrian flow is low, the proportion of the pedestrian term in the phase pressure of the PQ-MP, shown by Equation \eqref{eq:pressure}, is low. Therefore, under the control of PQ-MP, the phase phase serving pedestrians is activated less frequently in the low pedestrian demand region as well, and the delay reduction in this region is higher.

\subsubsection{Impact of measure errors on travel delays}
All results from the proposed PQ-MP in previous sections are based on accurate measurement of pedestrian queues; however, this requirement is difficult or even impossible to meet under the available infrastructures at most intersections. Therefore, it is necessary to test the performance of the PQ-MP in the presence of measurement errors in pedestrian queues. To this end, we assume that the measurement of a pedestrian queue for a crosswalk movement follows a normal distribution with mean equal to the actual pedestrian queue length and standard deviation equal to a proportion of the mean, denoted by $\sigma$. For example, if the actual pedestrian queue length is 10, and $\sigma=0.1$, the measurement output is obtained by randomly drawing a value from the distribution of $N(10,0.1\times 10)=N(10,1)$. Bounds are included to ensure that queues lengths are non-negative values. The following values for $\sigma$ are tested: $\{0.1, 0.2, 0.3, 0.4, 0.5\}$. 

Since we already demonstrated that both the rule-based MP and PQ-MP can effectively reduce person travel delay compared to Q-MP, for simplicity, we use the rule-based MP ($\tau=80 s$) with perfect knowledge of pedestrian waiting time as the baseline algorithm in this section. Figure \ref{fig:impact_delay} shows the influence of $\sigma$ on the delay reductions of vehicles, pedestrians and total persons. Figure \ref{fig:ped_improve_random_ped} shows that the pedestrian delay improvement from the PQ-MP generally decreases with the increase of $\sigma$. With the increase in the uncertainties in the measurement, the frequency of the PQ-MP makes the ``correct" decision for pedestrian queues, in terms of serving the phase with longest pedestrian queues, is lower, which leads to an undermining of the pedestrian operational efficiency. However, its influence on vehicle delays is less predictable, especially when the vehicle demand is high, as shown in Figure \ref{fig:veh_improve_random_ped}. A possible reason is that pedestrians have a higher priority than right-turn vehicles, which can generate either an increase or a decrease in vehicle delays. For example, if the phase for a pedestrian crosswalk movement that has a relatively short queue length is activated due to the measurement error, compared to the case when the controller activates the phase serving the longest pedestrian queue, the associated right-turning vehicles could incur less delay resulting from the shorter yielding time. Overall, when the vehicle demand delay is less than 600 veh/h/entry link, the PQ-MP with measurement randomness can still outperform the rule-based MP under all tested uncertainty values, even though the rule-based MP has perfect information. When the vehicle demand reaches 700 veh/h/entry link, as shown in Figure \ref{fig:total_improve}, both models with accurate measurement generate similar person delay. After considering measurement errors in the PQ-MP, the person delay increases and exceeds that from the rule-based MP under this vehicle demand level.

\begin{figure}[!ht]
    \centering
    \begin{subfigure}{0.8\textwidth}
        \includegraphics[width=4.5in]{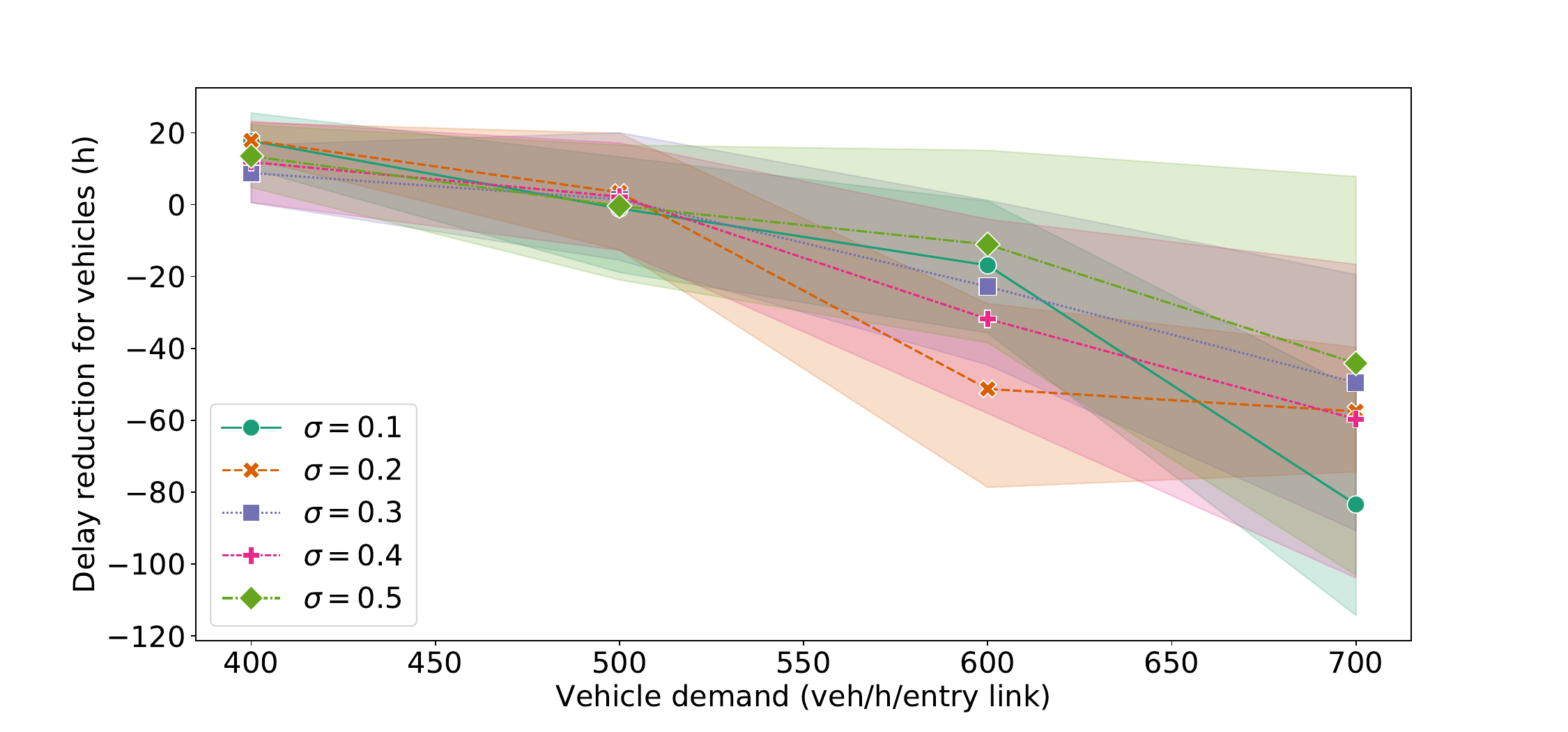}
    	\centering
    	\caption{Vehicle delay reduction.}
    	\label{fig:veh_improve_random_ped}
    \end{subfigure}
    
    \begin{subfigure}{0.8\textwidth}
        \includegraphics[width=4.5in]{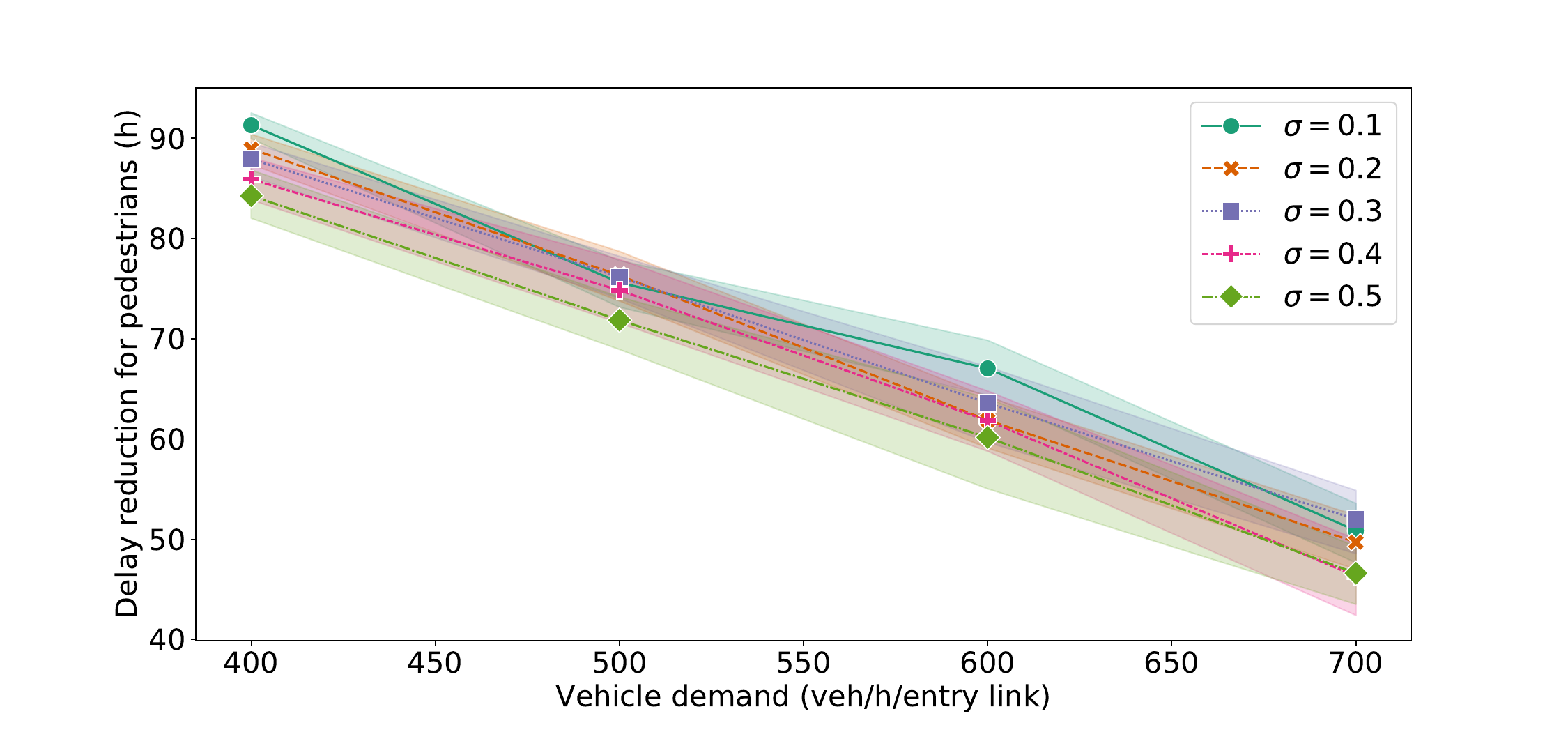}
    	\centering
    	\caption{Pedestrian delay reduction.}
    	\label{fig:ped_improve_random_ped}
    \end{subfigure}

    \begin{subfigure}{0.8\textwidth}
        \includegraphics[width=4.5in]{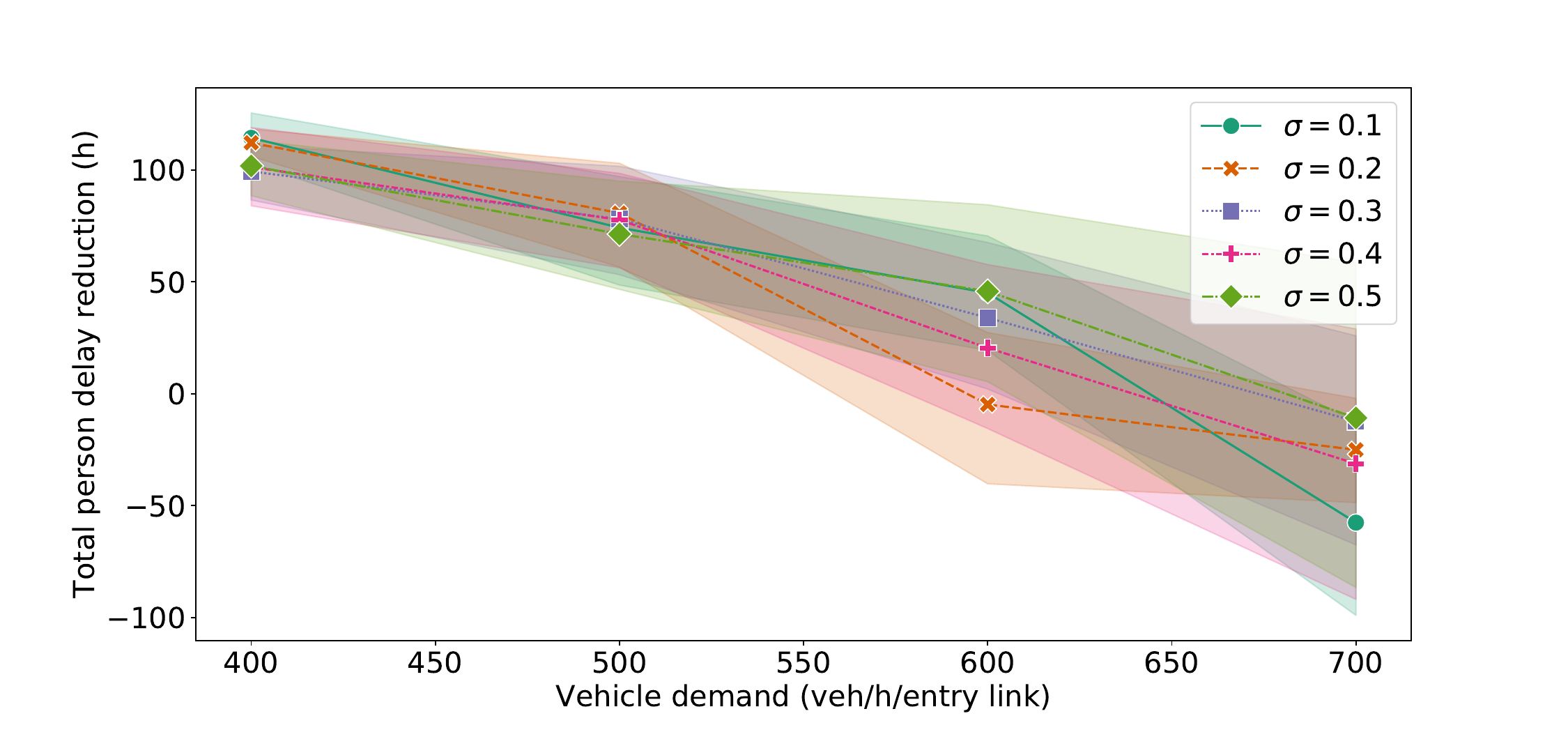}
    	\centering
    	\caption{Person delay reduction.}
    	\label{fig:total_improve_random_ped}
    \end{subfigure}
    \caption{Impact of measurement errors on delay.}
    \label{fig:impact_delay}
\end{figure}

\section{Concluding remarks}\label{sec:conclusion}
This paper proposes a novel MP algorithm, PQ-MP, that incorporates pedestrians into pressure computation to determine signal timings while maintaining the MP control structure. Pedestrian movements are divided into two groups: sidewalk movements and crosswalk movements, and traffic dynamics models are developed for both groups. Then, an MP algorithm that updates signal phase based on both vehicle queues and pedestrian queues of the crosswalk movement is derived. The proposed PQ-MP is decentralized, and does not require the knowledge on external demand. More importantly, we proved that PQ-MP has the maximum stability property, which ensures the maximization of the combined vehicle and pedestrian throughput. The stability in both vehicular and pedestrian queues is critical for the overall operational efficiency. This is highlighted by the simulation results which demonstrate that the PQ-MP generates a lower person delay than the Q-MP and a recently proposed rule-based MP algorithm at various vehicle demand levels, although the stable vehicle region of the PQ-MP is smaller than that of the rule-based model with a high threshold for pedestrian waiting times. We also tested the performance of the PQ-MP with measurement errors in pedestrian queues. The simulation results show that, in general, pedestrian delay increases with the increase in the measurement error. However, the impact of the measurement error on vehicle delays is less predictable than that on the pedestrian delay. Moreover, the PQ-MP with measurement errors still outperforms the rule-based model when the vehicle demand is not extremely large.

The proposed PQ-MP requires the number of pedestrians of each crosswalk movement at an intersection to be available, which may restrict its practicability. It is promising to propose a method to predict this metric based on readily available information, such as the first pedestrian arrival from pushing the cross button and the number of pedestrians served by each phase in previous steps, etc. In addition, since the maximum occupancy of pedestrian nodes is usually greater than that of vehicle links, it may require a longer time step than the free flow travel time for vehicles traversing a link to serve all pedestrians that are already waiting at the pedestrian nodes. Moreover, extending green time for T-R phases can help reduce the number of right-turn vehicles blocked by the conflicting pedestrian movements. Therefore, the development of an MP algorithm with dynamic time step size is another meaningful research direction. 

\section{Acknowledgements}
This research was supported by NSF, United States Grant CMMI-1749200.

\newpage

\bibliographystyle{elsarticle-harv}
\bibliography{biblio}
\end{document}